\documentclass[lettersize,journal]{IEEEtran}
\usepackage{amsmath,amsfonts}
\usepackage{algorithmic}
\usepackage[ruled, linesnumbered]{algorithm2e}
\usepackage{array}
\usepackage{textcomp}
\usepackage{stfloats}
\usepackage{url}
\usepackage{verbatim}
\usepackage{graphicx}
\usepackage{subcaption}
\usepackage{multirow}
\usepackage{cite}
\usepackage{amsthm}
\theoremstyle{remark}
\newtheorem{assumption}{Assumption}
\newtheorem{remark}{Remark}
\newtheorem{definition}{Definition}
\newtheorem{proposition}{Proposition}

\usepackage{amsmath}
\usepackage{booktabs}
\DeclareMathOperator*{\argmin}{arg\,min}
\usepackage{tikz}
\usetikzlibrary{arrows.meta, positioning, fit, backgrounds, calc}
\usetikzlibrary{shapes.geometric, positioning}

\begin{document}

\title{A General Set-Based Framework for Cognitive State Estimation: Theory and Application to Conditionally Automated Driving}

\author{\IEEEauthorblockN{Sibibalan Jeevanandam, Neera Jain}
\thanks{This material is based upon work supported by the National Science Foundation under Award No. 2145827. Any opinions, findings, and conclusions or recommendations expressed in this material are those of the authors and do not necessarily reflect the views of the National Science Foundation.}
\thanks{The authors are with the School of Mechanical Engineering, Purdue Unviersity, West Lafayette, IN 47906 USA (e-mail: sjeevana@purdue.edu; neerajain@purdue.edu) (\textit{Corresponding author: Neera Jain}).}
}



\maketitle

\begin{abstract} 
We present a set-based framework for estimating human cognitive states in human-automation interaction (HAI) contexts. Unlike probabilistic approaches dominant in the HAI literature, our framework treats process and measurement uncertainties as unknown but bounded, avoiding the need for large structured datasets or distributional assumptions on noise. We demonstrate the framework in the context of conditionally automated (SAE Level 3) driving, where we estimate three cognitive states that influence human reliance on the automation---trust, perceived risk, and workload---during a continuous, non-trial-based interaction. We leverage a hybrid dynamical modeling framework to identify individual-specific process and measurement models, systematically estimate noise bounds by enforcing reachset conformance, and identify the subset of cognitive states that influence each individual's reliance on the automation. Set-valued estimates of those states are then produced by fusing binary reliance observations and intermittent, quantized self-reports. The framework is evaluated through an in-person experiment in a medium-fidelity driving simulator with 20 participants. The set-valued estimator achieves at least 75\% consistency for most participants during testing, and multi-step-ahead reliance predictions derived from the estimates outperform both an open-loop baseline and a particle filter across all choices of prediction horizons (15, 30, 45, 60 time steps), with the performance gap more evident at longer horizons. The proposed estimation framework can enable automation systems that are continuously aware of, and responsive to, the human driver's state.
\end{abstract}


\section{Introduction}\label{sec:introduction}
Today, several commercially available automated vehicles (AVs) are equipped with conditional automation (SAE Level 3) which requires human supervision and intervention~\cite{shi_principles_2020}. The driver and the automation must therefore operate as a team, rendering automated driving a human-automation interaction (HAI) scenario. Safety is the primary objective during operation, yet it remains an open challenge. When the human driver bears full responsibility for safety, they tend to over-monitor the automation, sometimes experiencing a higher cognitive load than non-automated (manual) driving \cite{stapel_automated_2019}. Conversely, some drivers over-trust the automation and neglect the need to remain alert \cite{worle_misuse_2023}; in such cases, even safeguards in the vehicle, such as take-over requests, fail to increase safety. In both cases, forgoing automation altogether appears to be the safer option, defeating the purpose of AV development. However, both cases share a common assumption---that the automation operates without any knowledge of the driver's cognitive state. 

An alternative approach is to design adaptive automation to promote operational safety. This can be achieved by designing automation to respond to the human driver with personalized and timely interventions. For example, a driver whose trust exceeds the vehicle's actual reliability may require a longer take-over lead time. Conversely, the cognitive burden imposed by an opaque automation can be mitigated by selectively increasing system transparency \cite{akash_toward_2020}. Indeed, in several HAI contexts, closed-loop strategies that leverage the human's real-time cognitive state have led to measurable improvements in performance and safety \cite{saeidi_incorporating_2019,akash_human_2020,azevedo-sa_context-adaptive_2020,yuh_online_2024, wielatz_role_2025,chen_planning_2018, hu_dynamic_2024}. However, a necessary precursor to developing human-aware and adaptive AVs is a reliable algorithmic framework for real-time human cognitive state estimation.

\subsection{Literature Review}
Researchers have identified multiple cognitive factors that influence human behavior during automated driving. These include human trust in the automation \cite{sonoda_displaying_2017, korber_introduction_2018, stapel_-road_2022, he_modelling_2022}, their perception of risk\cite{stapel_-road_2022, he_modelling_2022}, and mental workload \cite{stapel_automated_2019}. 
\begin{table*}[t]
\centering
\caption{Recent literature on cognitive state estimation in human-AI interaction.}
\label{tab:lit_review}
\renewcommand{\arraystretch}{1.3}
\begin{tabular}{p{2cm} p{2.2cm} p{1.2cm} p{2.4cm} p{2.5cm} p{4cm} }
\hline
\textbf{Authors (Year)} & \textbf{Approach} & \textbf{State(s)} & \textbf{Representation} & \textbf{HAI Scenario} & \textbf{Observations / Inputs} \\
\hline
Xu \& Dudek (2015)~\cite{xu_optimo_2015} & Dynamic Bayesian Network  & Trust & Continuous-valued & Supervision of an aerial robot & Robot task performance, operator interventions, self-reports \\
Chen et al. (2018)~\cite{chen_planning_2018} & POMDP  & Trust & Discrete-valued & Teaming for table clearing & Human interventions, self-reports, performance\\
Wortelen et al.\ (2019)~\cite{wortelen_monte_2019} & Particle filter  & Workload & Continuous-valued & Partially automated driving & Perceivable driving-scene information and driver's actions\\
Akash et al. (2020)~\cite{akash_human_2020} & POMDP & Trust, Workload & Discrete-valued & Reconnaissance mission with automated decision aid & 
Human compliance, response time, scenario, automation transparency, performance\\
Azevedo-Sa et al. (2021)~\cite{azevedo-sa_real-time_2021} & Kalman Filter  & Trust & Continuous-valued & SAE Level 3 Automated Driving & Eye-tracking signals, system usage time, NDRT performance, event signals (true/false alarm, miss)
\\
Li et al. (2021)~\cite{li_kalman_2020} & Kalman filter  & Trust & Continuous-valued & Supervisory control of robotic swarms & Human response time, reliance, robots' task performance\\
Luo et al.\ (2023)~\cite{luo_real-time_2024} & Bayesian inference using ML classifiers & Workload & Discrete-valued & Tele-operation of a vehicle with a secondary visual search task & Gaze trajectory and pupil-size change \\
Williams et al. (2023)~\cite{williams_computational_2023} & POMDP  & Trust, Self-Confidence & Discrete-valued & Obstacle avoidance task with automation assistance & Human self-reports, reliance, performance\\
Caber et al.\ (2024)~\cite{caber_driver_2024} & Bayesian filter & Workload & Discrete-valued& On-road driving with in-vehicle system interactions & Driving-performance signals (steering, braking), road-type context\\
Hu et al.\ (2024)~\cite{hu_dynamic_2024} & Kalman filter  & Trust & Continuous-valued& SAE Level 2--3 automated driving & Driver's self-reported risk, attention rate, intervention rate, and NDRT completion rate \\
\hline
\textbf{This work} & Set-valued State Estimator & Trust, Risk, Workload & Set-Valued & SAE Level 3 automated-driving & Task complexity, driver self-reports, reliance \\
\hline
\end{tabular}
\end{table*}
Table~\ref{tab:lit_review} provides a summary of recent work on estimating these human cognitive states during HAI scenarios. Dynamic Bayesian Networks~\cite{xu_optimo_2015}, Partially Observable Markov Decision Processes (POMDPs)~\cite{chen_planning_2018,akash_human_2020,williams_computational_2023}, Kalman filters~\cite{azevedo-sa_real-time_2021,li_kalman_2020,hu_dynamic_2024}, particle filters~\cite{wortelen_monte_2019}, and other Bayesian estimators~\cite{caber_driver_2024,luo_real-time_2024} have been used to infer cognitive states such as trust, workload, and self-confidence. For example, Azevedo-Sa et al.~\cite{azevedo-sa_real-time_2021} use a Kalman filter to estimate driver trust during SAE Level 3 automated driving from eye-tracking signals, system usage, and event signals, modeling uncertainty as Gaussian process and measurement noise. Akash et al.~\cite{akash_human_2020} formulate a POMDP to jointly estimate trust and workload in a reconnaissance task, inferring discrete-valued cognitive states from human compliance, response time, and automation transparency cues. Wortelen et al.~\cite{wortelen_monte_2019} employ a particle filter to track continuous-valued driver workload during partially automated driving, using perceivable scene information and driver actions as observations. Caber et al.~\cite{caber_driver_2024} estimate discrete workload states during on-road driving using a Bayesian filter driven by driving-performance signals such as steering and braking, without relying on physiological or self-reported measurements.

Collectively, several observations emerge from Table~\ref{tab:lit_review}. First, these works establish probabilistic estimation as the dominant paradigm, requiring characterizations of process and measurement noise---typically as Gaussian~\cite{azevedo-sa_real-time_2021,li_kalman_2020,hu_dynamic_2024} or discrete conditional~\cite{xu_optimo_2015,chen_planning_2018,akash_human_2020,williams_computational_2023} distributions---whose parameters must be identified from large, structured datasets. Second, many approaches rely on frequent or periodic cognitive measurements collected over discrete interaction trials~\cite{xu_optimo_2015,chen_planning_2018,williams_computational_2023}. Automated driving violates both of these assumptions: sufficient data to identify individualized noise distributions is rarely available, and driver engagement with the automation is continuous rather than trial-based, so self-reports arrive intermittently rather than at regular intervals. Third, trust is the most commonly estimated cognitive state, appearing in nearly all prior work, with workload considered in a few studies~\cite{wortelen_monte_2019,akash_human_2020,luo_real-time_2024,caber_driver_2024}; risk perception---which studies have shown to influence reliance behavior during driving~\cite{stapel_-road_2022,he_modelling_2022, jeevanandam_hybrid_2025-1}---has not been directly estimated in any prior framework. These observations motivate a framework that (1) does not rely on a probabilistic noise model requiring large, structured datasets, (2) can operate with sparse, asynchronous cognitive measurements, and (3) can simultaneously estimate multiple cognitive states (e.g., trust, risk perception, and workload that influence the human user's behavior.

\subsection{Contribution}
We present a set-based framework for estimating human cognitive states in HAI contexts. A set-based characterization of model and measurement uncertainties avoids imposing probabilistic assumptions on the noise terms. Since human cognition can be influenced by numerous unmeasured factors, we believe bounding these factors---yielding a set-valued characterization of uncertainty---is more appropriate than a probabilistic one. Furthermore, set-based approaches have proven useful for state estimation when measurements are quantized\cite{sviestins_optimal_2000} or binary\cite{casini_set_2024}. Such measurements are common in HAI contexts; self-reports are typically solicited from the human driver on a scale with finite increments, yielding quantized measurements of the cognitive state. Moreover, reliance on the automation---whether or not the human uses the automation---constitutes a binary measurement of their cognition-related behavior. 

We demonstrate the framework in the context of conditionally automated driving for three cognitive states that influence human reliance on the automation during driving~\cite{jeevanandam_multi-factor_2026}---trust, perceived risk, and workload. We leverage our previously-developed hybrid dynamical modeling framework~\cite{jeevanandam_hybrid_2025} to identify individual-specific process and measurement models and assume the uncertainties to be unknown but bounded, systematically estimated by enforcing reachset conformance~\cite{gruber_scalable_2023}. For each individual, we identify the subset of cognitive states that influence their reliance and produce set-valued estimates of those states. This work represents the first set-based approach for human cognitive state estimation in a continuous, non-trial-based interaction with automation.

The remainder of this paper is organized as follows. We introduce some preliminaries and notation in Section~\ref{sec:prelim}. We present the process and measurement models and the underlying assumptions in Section~\ref{sec:model}. We describe the human user study used to illustrate the framework in Section~\ref{sec:human_study}. We elucidate model identification in Section~\ref{sec:sys_id} and describe the set-valued state estimator built on the identified model in Section~\ref{sec:state_estimator}. Finally we present and discuss the results from the human user study in Section~\ref{sec:results} and conclude the paper in Section~\ref{sec:conclusion}.
\section{Preliminaries and Notation}\label{sec:prelim}
Vectors are denoted by lowercase letters (e.g., $x$), unless specified otherwise. The $j$th standard basis vector in $\mathbb{R}^n$ is denoted by $e_j$, whose $j$th entry is $1$ and all other entries are $0$. For a vector $x\in\mathbb{R}^n$, we denote its $j$th coordinate using subscript $j$, such that $x_j := e_j^T x$. Similarly, for a matrix $A \in \mathbb {R}^{m\times n}$, $A_{ij}:=e_i^T Ae_j$. For any two vectors $x,y\in \mathbb R^n$, $|x|\leq y$ should be interpreted as $|x_j|\leq y_j \,\forall j=1,\cdots,n$. The $p$-norm of a vector $x\in \mathbb R^n$ is given by $\|x\|_p=\left(\sum^n_{j=1} |x_j|^p\right)^{1/p}$. Discrete time-indices are denoted by $k=0, \cdots, N$. 

Sets are denoted by uppercase calligraphic letters, e.g., $\mathcal{X}\subseteq \mathbb{R}^n$. A set raised to a positive exponent denotes the Cartesian product, such that $\mathcal{X}^m = \mathcal{X} \times \mathcal{X} \times \cdots \times \mathcal{X}$ ($m$ times). The linear image of $\mathcal{X}$ under $A\in\mathbb{R}^{m\times n}$ is $A\mathcal{X} := \{Ax : x\in\mathcal{X}\}\subseteq\mathbb{R}^m$. The Minkowski sum of any two sets $\mathcal{X},\mathcal{Y}\subseteq\mathbb{R}^n$ is $\mathcal{X}\oplus \mathcal{Y} = \{x+y :x\in\mathcal{X},\, y\in\mathcal{Y}\}$. The axis-aligned bounding box of $\mathcal X$, denoted by $\mathcal B(\mathcal X)$, is the smallest polytope with its faces parallel to the standard basis, that encloses all elements in $\mathcal X$. The axis-aligned volume of a set $\mathcal X$, denoted by $\text{vol}(\mathcal X)$, is defined as 
\begin{align*}
    \text{vol}(\mathcal X) = \prod_i^n \left(\rho(\mathcal B(\mathcal X),e_i)+\rho(\mathcal B(\mathcal X),-e_i)\right)\enspace,
\end{align*}
where $\rho(\mathcal Y,l)$ is the support function of a set $\mathcal{Y}\subseteq\mathbb{R}^n$ in direction $l\in\mathbb{R}^n$ such that
\(
    \rho({\mathcal{Y}},l) := \max \{\, l^\top y : y\in\mathcal{Y} \}.
\)

A vector, matrix, or set with subscript $\Pi$ represents its projection onto the subspace $\Pi \mathbb R^n$ defined using the transformation matrix $\Pi\in \mathbb R^{m\times n}$. For example, $A_\Pi=\Pi A$ and $\mathcal X_\Pi=\Pi\mathcal X$. 
Finally, set operations used in this work can result in non-convex polytopes; therefore, we use hybrid zonotopes~\cite{bird_hybrid_2023} to represent sets.

\section{Human Cognitive State Dynamics: Process and Measurement Models}\label{sec:model}
Consider the human cognitive state-space with $n$ states, each continuous-valued and taking values in the interval $[0,1]$. Accordingly, let $\mathcal S:=[0,1]^n$ denote the cognitive state-space with state vector $x \in \mathcal S$.

\subsection{Process Model}
The discrete-time evolution of the state vector $x$ can be described using a bounded affine difference equation \cite{jeevanandam_hybrid_2025}, such that
\begin{subequations}\label{eq:affine_dyn}
 \begin{align}
    x^+(k)&=Ax(k)+Bu(k)+ C + w(k)\enspace,\\
    x(k{+}1) &= \begin{cases}
        0 \text{, if }x^+(k)\leq0\\
        1 \text{, if }x^+(k)\geq1\\
        x^+(k) \text{, otherwise}\enspace.
    \end{cases}
\end{align}   
\end{subequations}
The model input $u \in \mathbb R^m$ can include both controllable factors such as the automation's availability and transparency, and uncontrollable but measurable factors such as task complexity (e.g., presence/absence of a construction zone). The system matrices for a specific individual are given by $A\in \mathbb R^{n\times n},B\in \mathbb R^{n\times m}, \text{and } C \in \mathbb R^{n}$. The process noise $w\in \mathbb R^n$ encapsulates the effects of unmodeled dynamics and unmeasured inputs.  
\begin{assumption}\label{assump:proc_noise}
    The process noise is unknown but bounded, such that $|w|\leq \epsilon_w$ for some $\epsilon_w \in \mathbb R^n$ ($\epsilon_w>0$).
\end{assumption} 
Accordingly, the set of all admissible values of $w$ is given by $\mathcal W:=\{w\in \mathbb R^n : |w|\leq \epsilon_w\}$.

\subsection{Measurement Models}
We consider two sensing modalities to infer the cognitive states---subjective (self-reported cognitive states) and behavioral (human reliance on the automation).

\subsubsection{Self-Reports}
The cognitive states can be measured intermittently through self-reports solicited from the human. Assuming each cognitive state can be reported on a scale of $0$ to $1$ in some fixed increment $\Delta\in (0,1)$, the self-reported cognitive state takes discrete values from the set $\mathcal Y:=\{0, \Delta, 2\Delta,\cdots,1\}$. Let $y\in \mathcal Y^n$ denote the output of the human's self-report of all three cognitive states at some time index $k\in \mathcal K_{SR}$. Here, $\mathcal K_{SR}$ denotes the set of time indices at which the human is solicited to self-report their cognitive state. Accordingly, $\forall  k\in \mathcal K_{SR}$, $y(k)$ is a measurement of $x(k)$, such that
\begin{subequations}\label{eq:sr_model}
\begin{align}
    s(k) &= x(k) + z(k)\enspace,\\
    y(k)&= [Q(s_1(k))\;\cdots\; Q(s_n(k))]^T \enspace.
\end{align}
\end{subequations}
 Here, $s\in \mathbb{R}^n$ is the latent, continuous-valued self-report that the human would provide if unconstrained by a discrete scale. The measurement noise $z\in \mathbb R^n$ captures effects of self-reporting bias \cite{rosenman_measuring_2011-1}, and $Q:\mathbb R \to \mathcal Y$ quantizes $s$ element-wise, such that for $i=1,\cdots,n$, 
\begin{align}\label{eq:quantizer}
    Q(s_i)=\begin{cases}
        0,\;\text{if } s_i< -\Delta/2\\
        1,\;\text{if } s_i\geq -\Delta/2\\
        \Delta\times \left\lfloor \frac{s_i}{\Delta} + \frac{1}{2} \right\rfloor,\; \text{otherwise}\enspace.
    \end{cases}
\end{align}
\begin{assumption}\label{assump:meas_noise}
    The measurement noise is unknown but bounded, such that $|z|\leq \epsilon_z$ for some $\epsilon_z \in \mathbb R^n$ ($\epsilon_z>0$).
\end{assumption}

Under Assumption~\ref{assump:meas_noise}, the set of all admissible values of $z$ is given by $\mathcal Z=\{z\in \mathbb R^n : |z|\leq \epsilon_z\}$.

\begin{definition}[Feasible State Set for Self-Reports]\label{def:FSS_SR}
    A feasible state set (FSS) for a self-report $y\in \mathcal Y$---denoted by $\mathcal X^{SR}_y$---is a set guaranteed to contain all states $x\in \mathcal S$ that can result in the observation $y$, such that $P\left(x(k)\in \mathcal{X}^{SR}_{y}\mid 
y(k)=y\right)=1$.
\end{definition}

Using Eqns.~\ref{eq:sr_model} and~\ref{eq:quantizer}, the feasible state set for a self-report $y$ is given by
\(
    \left(y\oplus \mathcal{Z}\oplus [-\Delta/2,\Delta/2)^n \right) \cap \mathcal S\enspace.
\)
For implementational convenience, we construct $\mathcal X^{SR}_{y}$ to be a closed set, such that
\begin{align*}
    \mathcal X^{SR}_y=\left(y\oplus \mathcal{Z}\oplus [-\Delta/2,\Delta/2]^n \right) \cap \mathcal S\enspace.
\end{align*}
Furthermore, without any loss of generality, we can combine the noise due to self-reporting bias and quantization, such that 
\begin{align*}
    \mathcal X^{SR}_y = \left(y\oplus \mathcal V\right)\cap \mathcal S\enspace,
\end{align*}
where $\mathcal V:=\{v\in \mathbb R^n:|v|\leq\epsilon_v =  \epsilon_z+\Delta/2\}$.
\vspace{0.5em}
\subsubsection{Human Reliance on the Automation}
At any time-index $0{\leq} k{\leq} N$, we can observe the human's behavior---their reliance on the automation---characterized by the binary variable $q(k)\in\{0,1\}$. Here, $q=1$ when the human relies on the automation, and $q=0$ when they are in manual control of the vehicle. 
The reliance behavior $q(k)$ is related to the cognitive state vector $x(k)$ through a function $h:\mathcal S\to \mathbb {R}$, such that the likelihood of the human relying on the automation is
\(
    p(q=1\mid x)=h(x)
\).
\begin{assumption}\label{assump:piecewise}
    The function $h$ is piecewise constant over $\mathcal S$, such that $h(x) = p_i$ when $x\in \mathcal S_i$ for some $i=1,\cdots,M$. Here, $\mathcal S_i \subseteq \mathcal S$ are sets that partition $\mathcal S$, such that $\bigcup_i^M \mathcal S_i=\mathcal S$, and $\mathcal S_i \cap \mathcal S_j = \emptyset$ for $i\neq j$.
\end{assumption}

\begin{definition}[Feasible State Set for Reliance]\label{def:FSS_reliance}
    A feasible state set for a reliance observation $q\in\{0,1\}$---denoted by $\mathcal X^R_q$---is a set guaranteed to contain all states $x\in \mathcal S$ that can result in the observed reliance $q$, such that $P\left(x(k)\in \mathcal{X}^R_{q}\mid 
q(k)=q\right)=1$.
\end{definition}


\section{Human User Study}\label{sec:human_study}
\begin{figure*}[t]
    \centering
    \begin{minipage}[t]{0.28\textwidth}
        \centering
        \includegraphics[width=\linewidth]{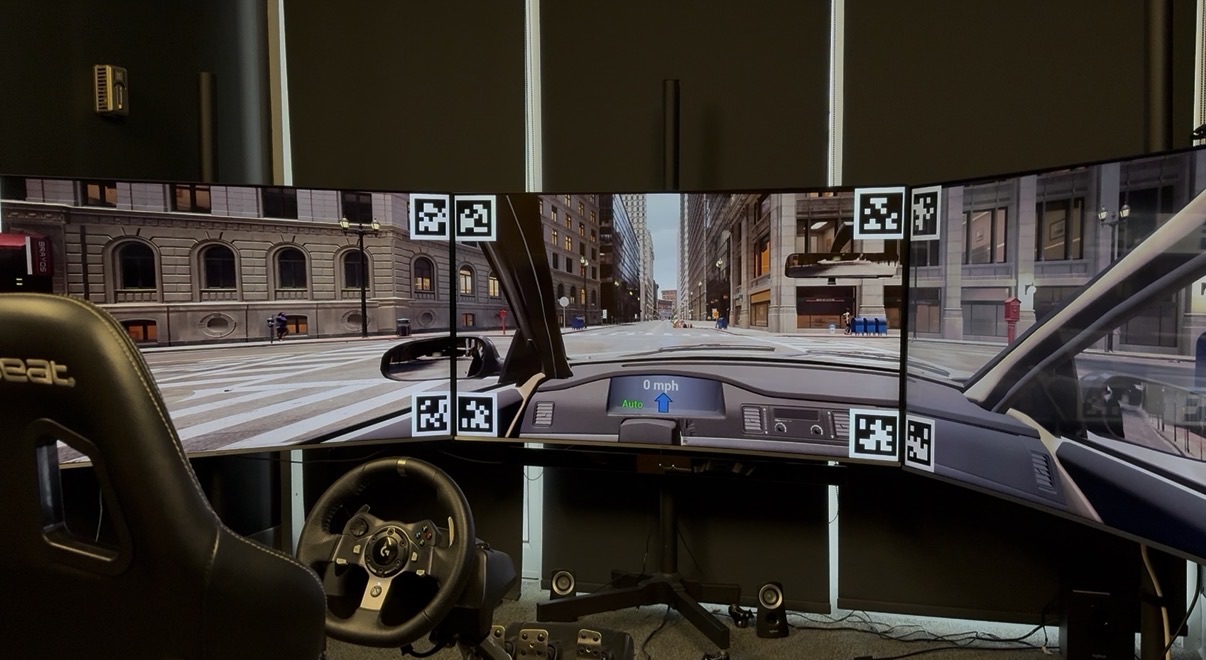}
        \caption{Driving simulator used in this study.}
        \label{fig:simulator}
    \end{minipage}
    \hfill
    \begin{minipage}[t]{0.68\textwidth}
        \centering
        \includegraphics[width=\linewidth]{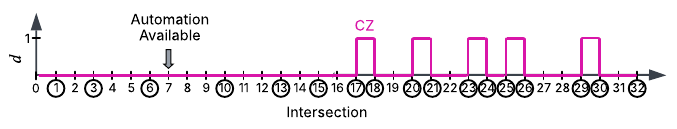}
        \caption{Binary signal representing task complexity during the main drive~\cite{jeevanandam_hybrid_2025-1}.}
        \label{fig:input_signal}
    \end{minipage}
\end{figure*}



\begin{figure}[t]
    \centering 
    \includegraphics[width=\linewidth]{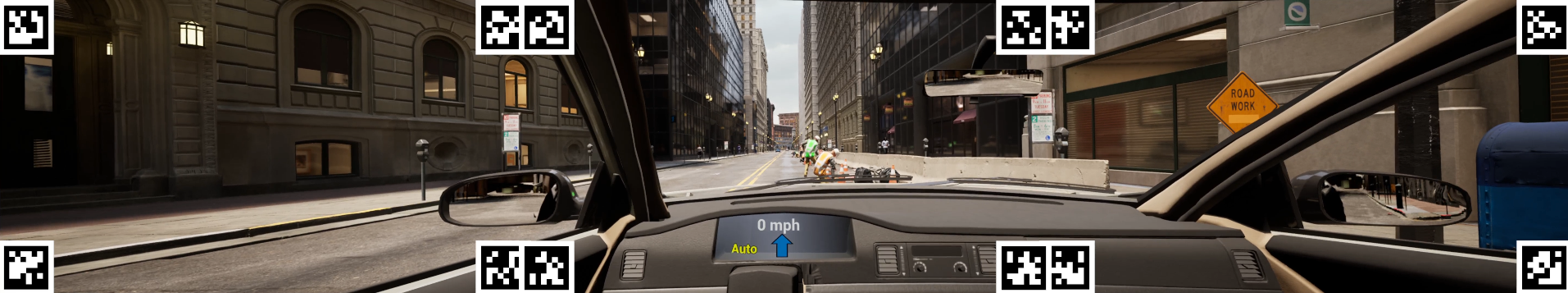}
    \caption{Participant's point-of-view when entering a construction zone (high task complexity scenario) in our driving simulator~\cite{jeevanandam_hybrid_2025-1}. }
  \label{fig:construction_pov}
\end{figure}

We designed an in-person experiment to elicit measurable changes in participants' cognitive states ($x$) and consequently their reliance on the automation ($q$)---see \cite{jeevanandam_hybrid_2025} for a detailed description of the experiment design and procedure. During the experiment, each participant navigates through a pre-defined route in a conditionally-automated vehicle with SAE Level 3 automation during a single, continuous drive in a custom-built medium fidelity driving simulator (Fig.~\ref{fig:simulator}). The input signal is defined as task complexity ($u:=d$) and is varied according to a binary sequence ($d\in\{0,1\}$) by placing construction zones along the route (see Fig.~\ref{fig:input_signal}). During low complexity ($d=0$), participants navigate through roads with low traffic, while during high complexity ($d=1$), participants are forced to merge into the left (opposite) lane and navigate around a construction zone (CZ) with human workers (Fig.~\ref{fig:construction_pov}). In this study, we solely manipulate task complexity; accordingly, the automation is capable of completing the entire route without errors (100\% reliable) albeit participants are not informed of this fact. They are instructed that once the automation is available to use, they can engage ($q=1$) or disengage ($q=0$) the automation at their will. Finally, at the circled intersections in Fig.~\ref{fig:input_signal}, participants are solicited to self-report (Table \ref{tab:self_reports}) their cognitive states when the ego vehicle stops at a stop sign or a red traffic light. In our prior work~\cite{jeevanandam_multi-factor_2026}, we found that three cognitive states---trust (T), perceived risk (R), and workload (W)---can influence human reliance on the automation during conditionally automated driving. Therefore, in this work, we restrict the cognitive state-space to be three-dimensional ($n=3$), such that $x=[T\; R\; W]^T$, and only measure these states. Participants respond on a scale ranging from 0 (labeled ``Very Low") to 100 (labeled ``Very High") in increments of 5. Without loss of generality, the self-reported cognitive states are normalized to range from $0$ to $1$, with the quantization interval $\Delta=0.05$.

\begin{table}[t]
    \centering
    \caption{Prompts for soliciting self-reported cognitive states~\cite{jeevanandam_hybrid_2025-1}.}
    \label{tab:self_reports}
    \begin{tabular}{l l}
        \toprule
        \textbf{State} & \textbf{Prompt} \\ \midrule
        $T$ & What is your current level of trust in the automation? \\
        $W$ & What is your current level of mental workload? \\
        $R$ & What is the current risk of an accident? \\ 
        \bottomrule
    \end{tabular}
    \vspace{-0.3cm}
\end{table}

\paragraph*{Participant Data} The study was approved by Purdue University's Institutional Review Board. Upon obtaining informed consent, 20 participants with a valid US driver's license (7 males and 13 females), aged between 19 and 33 (Mean=24.17, SD=4.35) participated in the study. The experiment lasted approximately one hour on average, and the participants were compensated for their time at a rate of \$5 per 15 minutes. Using the timestamps recorded for entering/exiting a construction zone, as well as engaging/disengaging the automation, the time series for task complexity ($d(k)$) and reliance ($q(k)$) is constructed for each participant with a sampling frequency of 1 Hz. We use $k=0$ to denote the discrete time index at which the ego vehicle arrives at intersection 15 in Fig.~\ref{fig:input_signal}, and $k=N$ to denote the end of the main drive (intersection 32). We start at intersection 15 to mitigate the effects of initial transients in the cognitive states soon after the automation is made available to the user at intersection 7. Finally, we round the timestamps recorded for self-reports to the nearest discrete-time indices ($kT_s$) and construct $\mathcal K_{SR}$. Thus, the following data is available for each participant: 
\begin{itemize}
    \item task complexity, $d(k)$ for $\;0\leq k\leq N$,
    \item reliance, $q(k)$ for $0\leq k\leq N$, and
    \item self-reports, $y(k)$ for $k \in \mathcal{K}_{SR}$.

\end{itemize}

\noindent Furthermore, we use $N_T$ to denote the time index corresponding to the self-report collected at intersection 24 (in Fig.~\ref{fig:input_signal}) and divide each participant's trajectory into two segments. The first segment ($0\leq k\leq N_T$) is used for system identification (Section~\ref{sec:sys_id}), and the second segment ($N_T\leq k\leq N$) is used to demonstrate state estimation (Section~\ref{sec:state_estimator}).

\section{System Identification}\label{sec:sys_id}
For each individual, the following unknowns need to be estimated: (i) the system matrices $A, B, C$ in Eqn.~\ref{eq:affine_dyn}, (ii) the noise bounds $\epsilon_w, \epsilon_v$, and (iii) the regions $\mathcal S_i$ and the associated probabilities $p_i$. These unknowns are estimated for each individual using their data for $0\leq k\leq N_T$. Algorithm~\ref{alg:sysid} summarizes the steps of the proposed system identification procedure.

\subsection{System Matrices}
Estimates of the system matrices, $\hat A, \hat B,\hat C$, are obtained by minimizing the sum of squared prediction errors with respect to the self-reports, such that 
\begin{subequations} \label{eq:sys_ID_matrices}
    \begin{align}
    &\hat A,\hat B,\hat C=\argmin_{A,B,C} \sum_{\substack{k \in \mathcal K_{SR}}}\left(\hat{x}(k) {-} y(k)\right)^2\\
    \text{s.t.  }
    &\hat x^+=\begin{cases}
        A\hat x(k)+ Bd(k) + C\;, k \notin \mathcal K_{SR} \\
        Ay(k)+ Bd(k) + C\;, k \in \mathcal K_{SR}
        \end{cases}\enspace,\label{eq:model_predict}\\
        &\hat x(k+1)=\max(\min(\hat x^+,1)),0),\;\hat x(0)=y(0)\enspace,\label{eq:model_init}\\
         &0\leq A_{ii}<1, A_{ij}=0, \forall i,j=1,2,3, i\neq j\enspace.\label{eq:A_const}
    \end{align}
\end{subequations}
Eqn.~\ref{eq:A_const} enforces a diagonal state matrix based on literature suggesting coupling between trust and cognitive workload is weak \cite{akash_reimagining_2020}. We also restrict $A$ to have non-negative entries under the assumption that the driver's state cannot oscillate in sign between consecutive time-indices under zero input. Furthermore, we restrict the diagonal entries of $A$ such that $A_{ii}<1$ under the assumption that the cognitive state dynamics are stable. The optimization problem described by Eqn.~\ref{eq:sys_ID_matrices} is solved in MATLAB using \textit{fmincon} \cite{the_mathworks_inc_optimization_2021}.

\subsection{Noise Bounds}
Given estimates of the system matrices ($\hat A, \hat B, \hat C$), we estimate the noise bounds, $\hat \epsilon_w, \hat \epsilon_v$, by adapting a reachset conformant identification procedure \cite{gruber_scalable_2023}. Specifically, we require that the set of states consistent with any self-report $y(k)$, given by $\mathcal X^{SR}_{y(k)}$, lie within the set of all possible states reachable by the model ($\hat A,\hat B,\hat C$) under all admissible $w\in \mathcal W$ at time-index $k$, denoted by $\mathcal X(k)$. Note that this can be trivially achieved by a suitably large $\hat \epsilon_w$. Therefore, we estimate the least conservative (smallest) noise bounds, such that
\begin{subequations}\label{eq:sys_ID_bounds}
    \begin{align}
    \hat \epsilon_v, \hat \epsilon_w&=\argmin_{\epsilon_v,\epsilon_w}\;  \|\epsilon_v\|_1+\| \epsilon_w\|_1 \\
    \text{s.t. } 
     &\mathcal X^+(k)= \begin{cases}
        \hat A\mathcal X(k)\oplus \hat B d(k) \oplus \hat C \oplus \mathcal W\;, k \notin \mathcal  K_{SR} \\
        \hat A \mathcal X_{y(k)}\oplus \hat B d(k) \oplus \hat C  \oplus \mathcal W\;, k \in \mathcal K_{SR}
    \end{cases},\label{eq:set_predict_1} \\
    &\mathcal X(k{+}1)=\mathcal X^+(k)\cap \mathcal S, \;\mathcal X(0)=\mathcal X^{SR}_{y(0)}\enspace, \label{eq:set_predict_2}\\
    &\mathcal X^{SR}_{y(k)}=\left(y(k)\oplus \mathcal V\right)\cap \mathcal S \subseteq \mathcal X(k) \;\forall k\in \mathcal K_{SR}\enspace,\label{eq:reachset_conformance}\\
    &\mathcal W=\{w\in \mathbb R^3:|w|\leq \epsilon_w\},\;\mathcal V=\{v\in\mathbb R^3: |v|\leq \epsilon_v\}\enspace,\\
    &\epsilon_v\geq\frac{\Delta}{2} \label{eq:eps_v_const}\enspace.
    \end{align}
\end{subequations}
Eqn.~\ref{eq:reachset_conformance} enforces reachset conformance. The measurement noise bound $\epsilon_v$ is bounded below by uncertainty due to quantization, enforced by Eqn.~\ref{eq:eps_v_const}. We simplify the optimization problem described by Eqn.~\ref{eq:sys_ID_bounds} (see details in Appendix~\ref{app:reachset_constraints}) and solve it in MATLAB using \textit{fmincon}.

\begin{algorithm}[t]
\DontPrintSemicolon
\caption{System Identification}
\label{alg:sysid}
\KwIn{Data $\{d(k), q(k), y(k)\}_{k=0}^{N_T}$; self-report indices $\mathcal{K}_{SR}$; quantization step $\Delta$}
\KwOut{Estimated parameters $\hat{A}, \hat{B}, \hat{C}$; noise bounds $\hat{\epsilon}_w, \hat{\epsilon}_v$; regions $\{\hat{\mathcal{S}}_i, \hat{p}_i, \hat{q}_i\}_{i=1}^{\hat{M}}$; projection $\Pi$; feasible sets $\mathcal{X}^R_q$}
\BlankLine

\tcp{Step 1: Estimate system matrices}
$\hat{A}, \hat{B}, \hat{C} \leftarrow \argmin_{A,B,C} \displaystyle\sum_{k \in \mathcal{K}_{SR}} \left(\hat{x}(k) - y(k)\right)^2$
\hfill subject to Eqns.~(\ref{eq:model_predict})--(\ref{eq:A_const})\;

\BlankLine
\tcp{Step 2: Estimate noise bounds}
$\hat{\epsilon}_v, \hat{\epsilon}_w \leftarrow \argmin_{\epsilon_v,\,\epsilon_w} \|\epsilon_v\|_1 + \|\epsilon_w\|_1$  subject to Eqns.~(\ref{eq:set_predict_1})--(\ref{eq:eps_v_const})\;

\BlankLine
\tcp{Step 3: Identify reliance regions via decision tree}
Compute state trajectory $\hat{x}(k)$ for $0 \leq k \leq N_T$ using $\hat{A}, \hat{B}, \hat{C}$ and Eqns.~(\ref{eq:model_predict})--(\ref{eq:model_init})\;
Fit binary decision tree with inputs $\hat{x}(k)$ and labels $q(k)$, max splits $= 2$\;
Extract regions $\hat{\mathcal{S}}_i$, class labels $\hat{q}_i$, and conditional probabilities $\hat{P}(q{=}1 \mid x \in \hat{\mathcal{S}}_i)$ for $i = 1, \ldots, \hat{M}$\;
$\hat{p}_i \leftarrow \hat{P}(q{=}1 \mid x \in \hat{\mathcal{S}}_i) \;/\; \mathrm{vol}(\hat{\mathcal{S}}_i)$ for each $i$

\BlankLine
\tcp{Step 4: Compute set-valued state trajectory}
Compute $\mathcal{X}(k)$ for $0 \leq k \leq N_T$ via Eqns.~(\ref{eq:set_predict_1})--(\ref{eq:set_predict_2}) using $\hat{A}, \hat{B}, \hat{C}, \hat{\epsilon}_w, \hat{\epsilon}_v$\;

\BlankLine
\tcp{Step 5: Construct feasible state sets for reliance}
\For{$i = 1, \ldots, \hat{M}$}{
    $\mathcal{L}_i \leftarrow \hat{\mathcal{S}}_i \;\cup \displaystyle\bigcup_{\substack{k \in \mathcal{K},\; q(k) = \hat{q}_i}} \mathcal{X}(k)$\;
}
$\mathcal{X}^R_q \leftarrow \displaystyle\bigcup_{i\,:\,\hat{q}_i = q} \mathcal{B}(\mathcal{L}_i)$ \quad for each $q \in \{0, 1\}$\;

\BlankLine
\tcp{Step 6: Identify reliance-influencing subspace}
Let $\mathcal{J} \subseteq \{1,2,3\}$ be the set of state indices appearing in at least one split node of the fitted tree\;
$\Pi \leftarrow \displaystyle\sum_{j \in \mathcal{J}} e_j e_j^\top$\;
\end{algorithm}
\subsection{Reliance Regions and Probabilities}
We estimate the number of reliance regions $\hat M$, the regions $\hat{\mathcal S}_i$, and their associated probabilities $\hat p_i$ (for $i=1,\cdots,\hat M$) via a binary decision tree \cite{breiman_classification_2017}. First, we use the estimated system matrices $\hat A,\hat B,\hat C$ and the measurements $y(k) \forall k\in\mathcal K_{SR}$ to estimate the state trajectory $\hat x(k) \text{ for } 0{\leq} k{\leq} N$ using Eqns.~\ref{eq:model_predict}-\ref{eq:model_init}. Then, a binary decision tree is fit with $\hat{x}(k)$ as the inputs and $q(k)$ as the class labels. The tree is trained in MATLAB using \textit{fitctree}~\cite{the_mathworks_inc_statistics_2021}, which uses the Standard CART (Classification and Regression Trees) algorithm \cite{breiman_classification_2017}. To prevent overfitting, the maximum number of splits is set to 2, thereby imposing $\hat M\leq 3$. The resulting tree defines the regions $\hat{\mathcal{S}}_i$, each associated with a conditional probability estimate $\hat{P}(q=1|x\in \hat{\mathcal{S}}_i)$. The regions are also assigned a predicted class label, denoted by $\hat q_i$, such that $\hat q_i=j$ if $\hat{P}(q=j|x\in \hat{\mathcal{S}}_i)\geq0.5$, where $j\in\{0,1\}$. The densities $\hat{p}_i$ are then recovered by normalizing the conditional probabilities by the volume of their respective regions $\hat{\mathcal{S}}_i$ (using Assumption~\ref{assump:piecewise}).

Next, we construct feasible state sets for reliance ($\mathcal X^R_q$) using the decision tree and the estimated system matrices and noise bounds. We first use the estimates of the system matrices ($\hat{A},\hat{B},\hat{C}$) and noise bounds 
($\hat{\epsilon}_w, \hat{\epsilon}_v$) to compute the set-valued state trajectory $\mathcal{X}(k)$ for $0{\leq} k{\leq} N$ via 
Eqns.~\ref{eq:set_predict_1}-\ref{eq:set_predict_2}. 
\
We then use $\hat{\mathcal S}_i$ to construct 
intermediate sets $\mathcal{L}_i$ for $i=1,\cdots,\hat{M}$, as
\begin{align*}
\mathcal{L}_i = \hat{\mathcal{S}}_i \bigcup_{\substack{k\in \mathcal{K} \\ q(k)=\hat{q}_i}} \mathcal{X}(k)\enspace,
\end{align*}
where $\hat{q}_i$ denotes the predicted class label for region 
$\hat{\mathcal{S}}_i$. 
Finally, we construct the feasible state set $\mathcal{X}^R_{q}$ as
\begin{align*}
\mathcal{X}^R_{q} = \bigcup_{i\,:\,\hat{q}_i=q} \mathcal{B}(\mathcal{L}_i)\enspace.
\end{align*}

\begin{proposition}\label{prop:FSS}
    The set $\mathcal{X}^R_{q}=\bigcup_{i\,:\,\hat{q}_i=q} \mathcal{B}(\mathcal{L}_i)$ is a feasible state set for the reliance observation $q$.
\end{proposition}
\begin{proof}
The set-valued state trajectory $\mathcal{X}(k)$ computed via Eqns.~\ref{eq:set_predict_1}-\ref{eq:set_predict_2} is guaranteed to contain the true state trajectory, such that $P(x(k)\in \mathcal{X}(k))=1$ for all $0\leq k \leq N_T$, due to Assumptions~\ref{assump:proc_noise}-\ref{assump:meas_noise} and constraints enforced by Eqns.~\ref{eq:set_predict_1}-\ref{eq:reachset_conformance}. 

Suppose $q(k) = q$ for some $k$. By construction,
\begin{align*}
    \mathcal X(k)\subset \underbrace{\mathcal L_{i} \subseteq \mathcal B(\mathcal L_{i})}_{\forall i:\hat q_i=q} \subseteq \bigcup_{i\,:\,\hat{q}_i=q} \mathcal{B}(\mathcal{L}_i)=\mathcal{X}^R_{q}\enspace,
\end{align*}
and thus $P\left(x(k) \in \mathcal{X}^R_{q}\right)=1$. This is true for all $0\leq k\leq N_T$ at which $q(k)=q$. Therefore, $P\left(x(k) \in \mathcal{X}^R_{q}\mid q(k)=q\right)=1$, and thus $\mathcal{X}^R_{q}=\bigcup_{i\,:\,\hat{q}_i=q} \mathcal{B}(\mathcal{L}_i)$ is a feasible state set (by Definition~\ref{def:FSS_reliance}).
\end{proof}

From the decision tree, we can also identify the cognitive states that influence reliance. While all three cognitive states have the potential to influence reliance, in our prior work~\cite{jeevanandam_hybrid_2025} we found that the subset of cognitive states influencing reliance on the automation varies across individuals. Therefore, in this work, we use splits in the decision tree to identify the subspace that influences reliance. Specifically, each split in the binary decision tree partitions the state space along a single state dimension $x_j$, using a threshold. Let $\mathcal J\subseteq \{1, 2, 3\}$ denote the set of state indices $j$ that appear in at least one split node of the fitted tree. Then, the subspace influencing reliance can be characterized using the projection matrix
\begin{align*}
    \Pi = \sum_{j \in \mathcal{J}} e_j e_j^T\enspace.
\end{align*}
 The projected state $\tilde{x} = \Pi x$ retains the elements in $x$ that are predictive of operator reliance, excluding the dimensions that are not. Identifying such a subspace enables dimensionality reduction in the set-valued state estimator; during a real-time implementation, only the cognitive state dimensions in $\mathcal{J}$ need to be solicited via self-reports, potentially reducing user workload due to querying. Note that in this work, since reliance is the only measurement available at all time indices, the uncertainty in states outside the subspace defined by $\mathcal J$ would, in general, grow over time and could only be reduced when self-reports are collected.  Incorporating measurements that are influenced by all cognitive state dimensions to enable full-state estimation is left to future work.

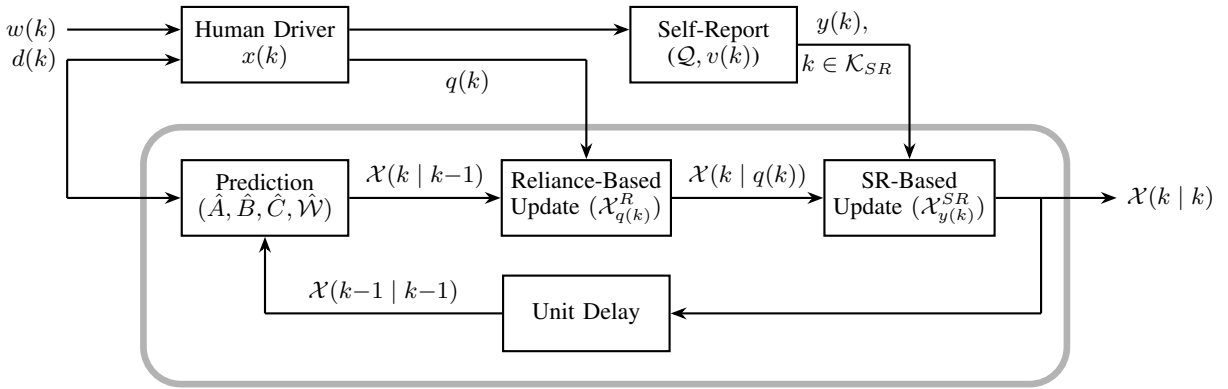
\begin{figure*}[b]
\centering
\begin{tikzpicture}[
    block/.style={
        draw, rectangle, thick,
        minimum width=2.2cm, minimum height=1.0cm,
        align=center, font=\small
    },
    arr/.style={-{Stealth[length=6pt]}, thick},
    line/.style={thick},
]
 
\node[block] (driver)   {Human Driver \\ $x(k)$};
\node[block, right=3.7cm of driver] (selfrep) {Self-Report\\($\mathcal Q, v(k)$)};
 
\node[block, below=1.0cm of driver, xshift=0cm] (pred)
    {Prediction\\  $(\hat{A},\hat{B},\hat{C},\hat{\mathcal{W}})$};
\node[block, right=2.0cm of pred]     (reliance)  {Reliance-Based \\ Update $(\mathcal{X}^R_{q(k)})$};
\node[block, right=2.0cm of reliance] (srupdate)  {SR-Based\\ Update $(\mathcal{X}^{SR}_{y(k)})$};
\node[block, below=0.5cm of reliance] (delay)     {Unit Delay};
 
\node[inner sep=0pt, right=0.4cm of srupdate] (rightpad) {};
\begin{scope}[on background layer]
    \node[draw=gray!60, line width=2.5pt, rounded corners=14pt,
          fit=(pred)(reliance)(srupdate)(delay)(rightpad),
          inner xsep=14pt, inner ysep=12pt] (innerbox) {};
\end{scope}
 
\draw[arr] ([xshift=-1.5cm, yshift= 0.2cm] driver.west)
    node[left,font=\small]{$w(k)$} -- ([yshift= 0.2cm] driver.west);
\draw[arr] ([xshift=-1.5cm, yshift=-0.2cm] driver.west)
    node[left,font=\small]{$d(k)$} -- ([yshift=-0.2cm] driver.west);
 
\coordinate (dIn)   at ([xshift=-1.5cm, yshift=-0.2cm] driver.west);
\coordinate (dDown) at (dIn |- pred.west);
\draw[line] (dIn)   -- (dDown);
\draw[arr]  (dDown) -- (pred.west);
 
\draw[arr] ([yshift=0.2cm] driver.east) -- ([yshift=0.2cm] selfrep.west);
 
\coordinate (qExit) at ([yshift=-0.2cm] driver.east);
\coordinate (qDrop) at (qExit -| reliance.north);
\draw[line] (qExit) -- node[below, font=\small]{$q(k)$} (qDrop);
\draw[arr]  (qDrop) -- (reliance.north);
 
\coordinate (srExit) at (selfrep.east);
\coordinate (srDrop) at (srExit -| srupdate.north);
\draw[line] (srExit) -- node[above, font=\small, pos=0.4]{$y(k),$} node[below, font=\small, pos=0.3]{$\;\;\;\;\,k \in \mathcal{K}_{SR}$} (srDrop);
\draw[arr]  (srDrop) -- (srupdate.north);
 
\draw[arr] (pred.east) --
    node[above, font=\small]{$\mathcal{X}(k\mid k{-}1)$}
    (reliance.west);
 
\draw[arr] (reliance.east) --
    node[above, font=\small]{$\mathcal{X}(k\mid q(k))$}
    (srupdate.west);
 
\coordinate (outStart) at (srupdate.east);
\coordinate (branch)   at ([xshift=0.6cm] outStart);
\coordinate (delayIn)  at (branch |- delay.east);
 
\draw[arr] (outStart) -- ++(1.6,0)
    node[right, font=\small]{$\mathcal{X}(k\mid k)$};
 
\draw[line] (branch)  -- (delayIn);
\draw[arr]  (delayIn) -- (delay.east);
 
\coordinate (corner) at (delay.west -| pred.south);
\draw[line] (delay.west) --
    node[above, font=\small]{$\mathcal{X}(k{-}1\mid k{-}1)$}
    (corner);
\draw[arr]  (corner) -- (pred.south);
 
\end{tikzpicture}
\caption{Proposed set-valued state estimation architecture. The gray box highlights the set-valued state estimator, whose inputs are the measured disturbance ($d(k)$), observed reliance ($q(k)$), and cognitive state self-reports ($y(k)$).}
\label{fig:svse}
\end{figure*}

 \begin{algorithm}[t]
\DontPrintSemicolon
\caption{Set-Valued State Estimation}
\label{alg:svse}
\KwIn{Estimated parameters $\hat{A}, \hat{B}, \hat{C}$; noise sets $\hat{\mathcal W},\hat{\mathcal V}$; projection $\Pi$; feasible sets $\mathcal{X}^R_q$; window $n_q$; data $\{d(k), q(k), y(k)\}$}
\KwOut{Set-valued state estimates $\{\mathcal{X}(k \mid k)\}_{k=N_T}^{N}$}
\BlankLine
$\mathcal{X}(N_T|N_T) \leftarrow \Pi\, \mathcal{X}^{SR}_{y(N_T)}$ \tcp*[r]{Initialization}
$\mathcal X\leftarrow \mathcal{X}(N_T|N_T)$\tcp*[r]{Temporary variable}

$\texttt{flag}_\text{retrain} \leftarrow \texttt{false}$\tcp*[r]{If true, re-train}
\BlankLine
\For{$k = N_T + 1, \ldots, N$}{
    $\mathcal{X} \leftarrow (\hat{A}_\Pi\, \mathcal{X} \oplus \hat{B}_\Pi d(k) \oplus \hat{C}_\Pi \oplus \hat{\mathcal{W}}_\Pi)\cap \Pi \mathcal S$
    \tcp*[r]{One-step-ahead prediction}
    $\mathcal M=\Pi \mathcal X^R_{q(k)}$
    
    \For{$m=1:\min(n_q,k{-}N_T)$}{
    $\mathcal M(k{-}m|k{-}m)=\mathcal X^R_{q(k{-}m)}$
    
    \For{$j=1:m$}{$\mathcal M(k{-m}{+}j| k{-}m)\leftarrow \hat{A}\, \mathcal M(k{-m}{+}j{-}1| k{-}m) \oplus \hat{B} d(k{-}m{+}j) \oplus \hat{C} \oplus \hat{\mathcal{W}}$}
    $\mathcal M\leftarrow\mathcal M \cap \Pi\mathcal M(k|k{-}m)$
    }
    \tcp*[r]{Propagate past reliance FSS}
    $\mathcal{X}_\text{rel} \leftarrow \mathcal{B}(\mathcal{X}) \cap \mathcal{M}$
    \tcp*[r]{Update using $q(k)$}
    \uIf{$\mathcal{X}_\text{rel} = \emptyset$}{
        $\mathcal{X} \leftarrow \mathcal{B}\!\left(\mathcal{X} \cup \Pi \mathcal X^R_{q(k)}\right)$\;
        $\texttt{flag}_\text{retrain} \leftarrow \texttt{true}$\;
    }
    
    \If{$k \in \mathcal{K}_{SR}$}{
        $\mathcal{X} \leftarrow \mathcal{X} \cap \Pi\, \mathcal{X}^{SR}_{y(k)}$
        \tcp*[r]{Update using $y(k)$}
        \uIf{$\mathcal{X} = \emptyset$}{
            $\texttt{flag}_\text{retrain} \leftarrow \texttt{true}$
        }
            
        \If{$\texttt{flag}_\text{retrain}$}{
        $k_o\leftarrow k$\;
            Re-estimate $\hat{A}, \hat{B}, \hat{C}, \hat{\mathcal W}, \hat{\mathcal V}, \Pi, \mathcal X^{R}_q$ with data for $0 \leq k \leq k_o$ using Algorithm~\ref{alg:sysid}
            
            $\mathcal{X} \leftarrow \Pi\, \mathcal{X}^{SR}_{y(k)}$\;
            $\texttt{flag}_\text{retrain} \leftarrow \texttt{false}$\;
        }\tcp*[r]{Online model update}
    }
    $\mathcal{X}(k \mid k) \leftarrow \mathcal{X}$\;
}
\end{algorithm}

\section{Set-Valued State Estimation}\label{sec:state_estimator}
In this section, we describe our set-valued state estimation scheme, illustrated in Fig.~\ref{fig:svse}. Upon identifying the system matrices, noise bounds, and reliance regions for an individual (Algorithm~\ref{alg:sysid}), we estimate their set-valued cognitive states for $N_T\leq k\leq N$ in the subspace $\mathcal S_\Pi$. Each step of the estimation algorithm is described below and then summarized as Algorithm~\ref{alg:svse}. We use the zonoLAB~\cite{koeln_zonolab_2024} toolbox in MATLAB to define and compute sets.
 
\subsubsection{Initialization}
The state estimator is initialized at time-index $k=N_T$ using the self-report $y(N_T)$, such that 
\begin{align*}
    \mathcal X(N_T\mid N_T)=\Pi \mathcal X^{SR}_{y(N_T)}=(y(N_T)_\Pi\oplus \hat{\mathcal V}_\Pi)\cap\mathcal S_\Pi,
\end{align*}
where $\hat{\mathcal V}=\{v\in\mathbb R^3: |v|\leq \hat \epsilon_v\}$.
\subsubsection{One-Step-Ahead Set-Valued Prediction}
For subsequent time indices $k{>}N_T$, we compute the set of all possible states reachable in one time step from $\mathcal X(k{-}1\mid k{-}1)$ via the estimated model (characterized by $\hat A,\hat B,\hat C$) under the disturbance $d(k)$ and process noise $w(k)\in \hat{\mathcal W}=\{w\in \mathbb R^3:|w|\leq \hat \epsilon_w\}$. This reachable set, denoted by $\mathcal X(k\mid k{-}1)$, is computed as
\begin{align}\label{eq:one_step_ahead}
    \mathcal X(k\mid k{-}1)=(\hat A_\Pi\mathcal X(k{-}1\mid k{-}1)\oplus \hat B_\Pi d(k-1) \oplus &\hat C_\Pi  \oplus \hat{\mathcal W}_\Pi) \notag\\ &\cap \mathcal S_\Pi \enspace.
\end{align}

\subsubsection{Set-Update Using Observed Reliance}
At time-index $k$, upon observing the human's reliance on the automation $q(k)$, the set-valued state estimate is updated by computing the intersection between $\mathcal X(k\mid k{-}1)$ and the feasible state set for $q(k)$, such that
\begin{align*}
    \mathcal X(k\mid q(k))=\mathcal X(k\mid k{-}1)\cap \Pi\mathcal X^R_{q(k)}\enspace.
\end{align*}
However, due to the conservativeness of $\mathcal X^{R}_{q(k)}$, it may be beneficial to simultaneously process a block of past measurements $q(k)$ \cite{casini_set_2024}. Following the procedure described in \cite{casini_set_2024}, we define $\mathcal M(k\mid k{-}m)$ to be the $m$-step ahead propagation of $\mathcal X^R_{q(k{-}m)}$ for $k>m$, such that 
\begin{align*}
    \mathcal M(k\mid k{-}m)=A^m\mathcal X^R_{q(k{-}m)} \oplus \sum_{j=0}^{m-1}A^{m-1-j}Bd(k{-}m{+}j)\oplus \\\left(\sum_{j=1}^{m} \hat A^{m-j}\right)\hat C \oplus \bigoplus_{j=1}^{m} \hat A^{m-j}\hat{\mathcal W}\enspace.
\end{align*}
We then use the past $N_q$ observations of reliance to update the set-valued state estimate, such that
\begin{align}\label{eq:update_rel}
    \mathcal X(k\mid q(k))=\mathcal B(\mathcal X(k\mid k{-}1))\bigcap_{m=0}^{N_q} \Pi\mathcal M(k\mid k-m)\enspace.
\end{align}
Here, $\mathcal B(\mathcal X(k\mid k{-}1))$ is used as an over-approximation to limit the complexity of $\mathcal X(k\mid q(k))$ due to recursive intersection operations. The parameter $N_q$ can be used to trade off computational complexity and specificity of the set-estimate. In this work, $N_q=10$.
\subsubsection{Set-Update Using Self-Reports}
At some $k\in \mathcal K_{SR}$, we use the feasible state set for the self-reported cognitive state vector $y(k)$ to update the set-valued state estimate, such that
\begin{align}\label{eq:update_sr}
    \mathcal X(k\mid q(k),y(k))=\mathcal X(k\mid q(k))\cap \Pi \mathcal X^{SR}_{y(k)}\enspace.
\end{align}

\subsubsection{Online Model Update}
The set-valued state estimates $\mathcal X(k\mid k)$ are guaranteed to contain the projected state $\tilde{x}(k)=\Pi x(k)$, under the assumption that the training data has captured all behaviors characteristic of the human user. However, this assumption is unrealistic---the system matrices, noise bounds, and reliance regions estimated using the training set may not capture human behaviors observed in the testing set. Due to such discrepancies, the set-update steps (Eqns.~\ref{eq:update_rel} and \ref{eq:update_sr}) can yield empty sets, which we interpret as a failure of the parameter estimates to guarantee containment of the true state. Thus, if at some $k=k_d$,
\(
    \mathcal X(k\mid q(k))\;\text{or}\;\mathcal X(k\mid q(k),y(k))
\)
is empty, we re-train the model at the next available self-report (e.g. at time-index $k_o\in \mathcal K_{SR}$, $k_o\geq k_d$). Furthermore, to continue producing valid state estimates for $k_d\leq k\leq k_o$, we let $\mathcal X(k\mid k)$ be an over-approximated set, such that
\begin{align*}
    \mathcal X(k\mid k)=\begin{cases}
        \mathcal X(k\mid q(k)),\; \text{if } \mathcal X(k\mid q(k))\neq \emptyset\\
        \mathcal B\left(\mathcal X(k\mid k{-}1)\cup \Pi\mathcal X^R_{q(k)}\right), \;\text{otherwise}\enspace.
    \end{cases}
\end{align*}

At $k=k_o$, we re-estimate the system matrices, noise bounds, and reliance regions using Algorithm~\ref{alg:sysid} with data for $0\leq k\leq k_o$. The state estimate is reset at $k=k_o$ using $y(k)$, such that $\mathcal X(k\mid k)=\Pi \mathcal X^{SR}_{y(k_o)}$.
\begin{remark}
    The subspace influencing reliance may change during the online model update step; we update the projection matrix $\Pi$ based on the re-trained decision tree, and compute set-valued state estimates in the (potentially) new subspace for subsequent time indices.
\end{remark}

\section{Results and Discussion}\label{sec:results}
For each participant, we apply Algorithm~\ref{alg:sysid} to estimate their system matrices, noise bounds, and reliance regions using the procedure outlined in Section~\ref{sec:sys_id}. We do not analyze data from 5 participants: P02, P11, and P16 did not change their reliance during the training segment ($0\leq k\leq N_T$), precluding decision tree identification, and P17 and P19 crashed into a vehicle or road object during the experiment. For the remaining 15 participants, we apply Algorithm~\ref{alg:svse} to compute their set-valued state estimates for $N_T\leq k\leq N$.

Fig.~\ref{fig:P017_combined} illustrates the identified decision tree, feasible state sets for reliance, and set-updates at time index $k=244$ (first self-report in testing phase) for participant P01. The decision tree structure (Fig.~\ref{fig:P017_tree}) indicates that the participant's reliance is influenced by two of the three cognitive states---risk perception and trust. Specifically, they are more likely to rely on the automation when their perception of risk ($R$) is less than $0.3$ and their trust ($T$) is greater than $0.65$. Fig.~\ref{fig:FSS_P017} illustrates the feasible state sets $\mathcal{X}_0^R$ and  $\mathcal{X}_1^R$. Recall that these sets are not estimates of the decision tree partitions $\{\hat{\mathcal{S}}_i\}$, but rather guaranteed outer approximations 
of the true decision regions (Proposition~\ref{prop:FSS}). As a consequence, unlike the sets $\{\hat{\mathcal{S}}_i\}$ defined by the decision tree, the sets $\mathcal{X}_0^R$ and $\mathcal{X}_1^R$ need not be disjoint---the crosshatched region in Fig.~\ref{fig:FSS_P017} ($\mathcal{X}_0^R \cap \mathcal{X}_1^R$) illustrates the overlap between the two regions due to conservatism. Fig.~\ref{fig:P017_sets} shows the sets updated using reliance ($\mathcal X(k \mid q(k))$), and the feasible set for the self-report ($\mathcal X^{SR}_{y(k)}$) collected at time index $k=244$; the feasible set is contained ($\mathcal X^{SR}_{y(k)}\subset \mathcal X(k \mid q(k))$), indicating reachset conformance during the testing phase.   

\begin{figure*}
    \centering
    \begin{subfigure}[b]{0.3\linewidth}
        \centering
        \scalebox{0.7}{\begin{tikzpicture}[
    every node/.style={font=\large},
    branch/.style={draw, regular polygon, regular polygon sides=3, inner sep=1pt, blue},
    leaf/.style={draw, circle, fill=blue, inner sep=3pt, blue},
    edge/.style={blue}
]

\node[branch] (root) at (4, 0) {};

\node[branch] (left) at (1.5, -2.5) {};

\node[leaf] (right) at (6.5, -2.5) {};
\node[below=4pt of right] {$\hat{q}=0$};

\node[leaf] (leftleft) at (0, -5) {};
\node[below=4pt of leftleft] {$\hat q=0$};

\node[leaf] (leftright) at (3, -5) {};
\node[below=4pt of leftright] {$\hat q=1$};

\draw[edge] (root) -- node[above left, black] {$R < 0.30$} (left);
\draw[edge] (root) -- node[above right, black] {$R \geq 0.30$} (right);
\draw[edge] (left) -- node[above left, black] {$T < 0.65$} (leftleft);
\draw[edge] (left) -- node[above right, black] {$T \geq 0.65$} (leftright);

\end{tikzpicture}}
        \caption{Decision tree splits identified using training data; reliance is influenced by trust and risk perception.}
        \label{fig:P017_tree}
    \end{subfigure}
    \hfill
    \begin{subfigure}[b]{0.3\linewidth}
        \centering
        \includegraphics[width=\linewidth]{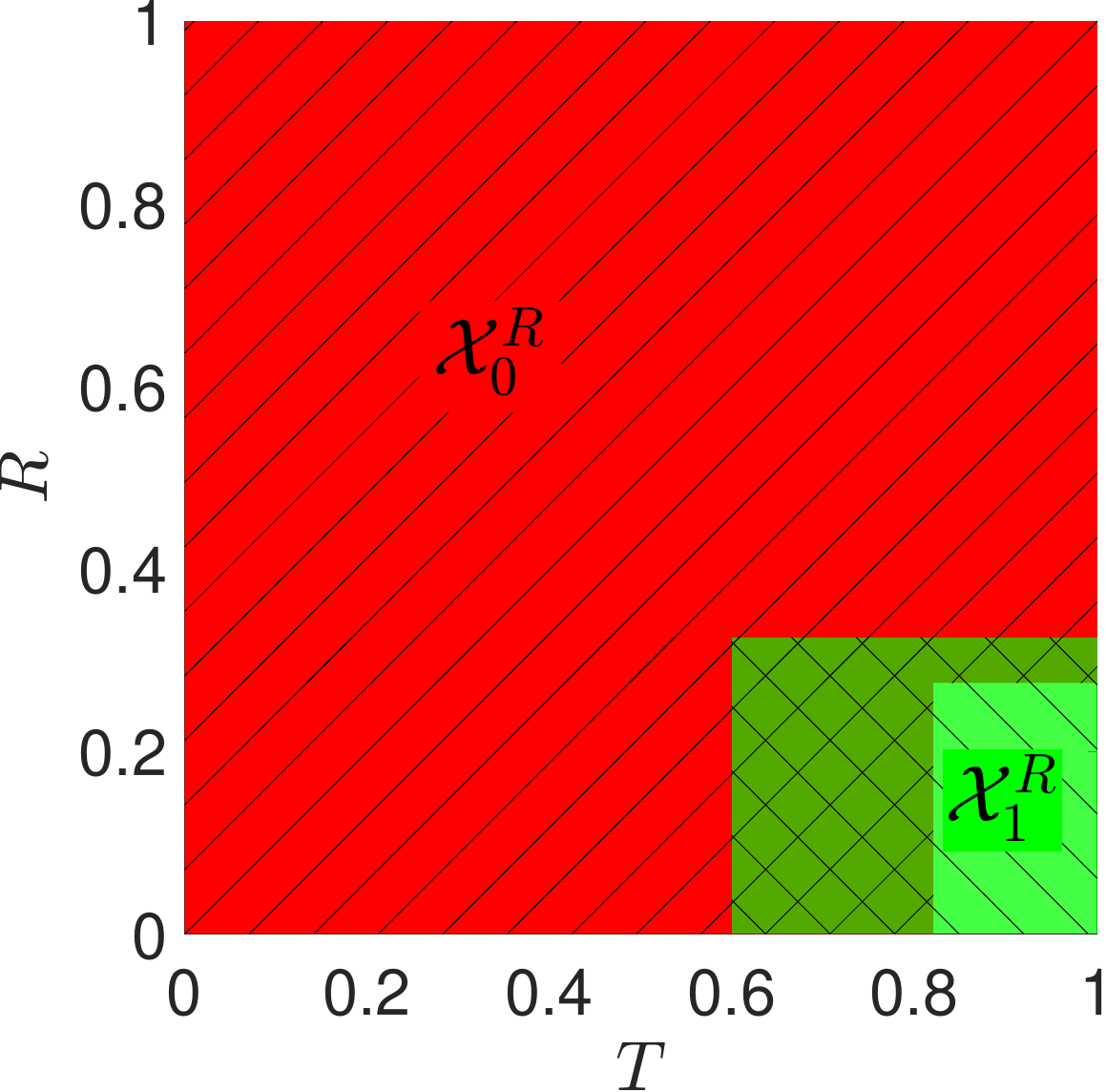}
        \caption{Feasible state sets for reliance; the sets $\Pi \mathcal X^{R}_0$ (red, $45^\circ$ hatch) and $\Pi \mathcal X^{R}_1$ (green, $135^\circ$ hatch) have a non-empty intersection (dark green, crosshatched).}
        \label{fig:FSS_P017}
    \end{subfigure}
    \hfill
    \begin{subfigure}[b]{0.3\linewidth}
        \centering
        \includegraphics[width=\linewidth]{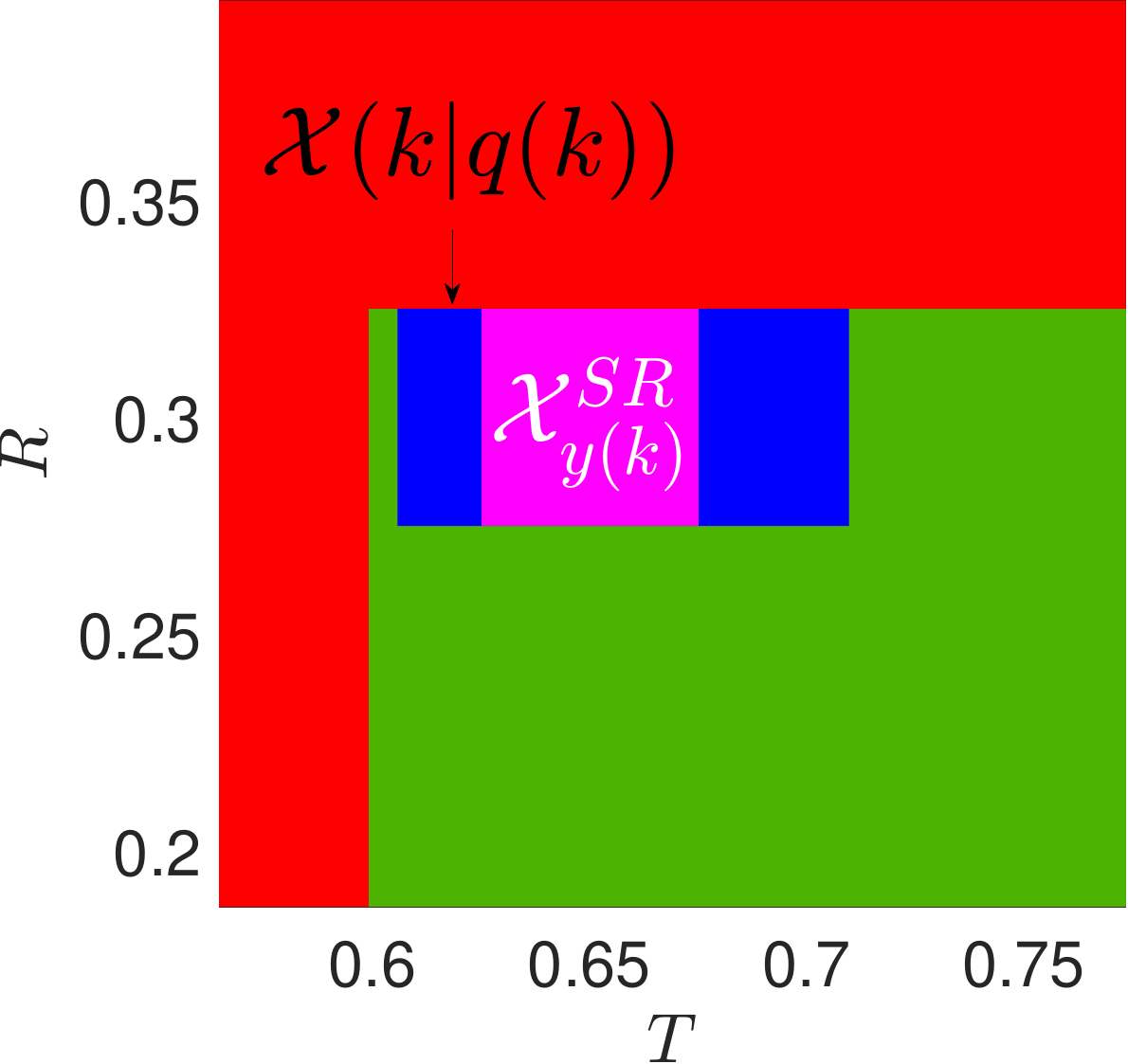}
        \caption{Zoomed-in view of reliance-updated set $\mathcal{X}(k|q(k))$ (blue) and self-report feasible set $\Pi \mathcal{X}^{SR}_{y(k)}$ (magenta) at $k=244$; the feasible set (magenta) is contained within $\mathcal{X}(k|q(k))$ (blue).}
        \label{fig:P017_sets}
    \end{subfigure}
    \caption{The identified decision tree, feasible state sets for reliance, and set-updates for participant P01.}
    \label{fig:P017_combined}
\end{figure*}
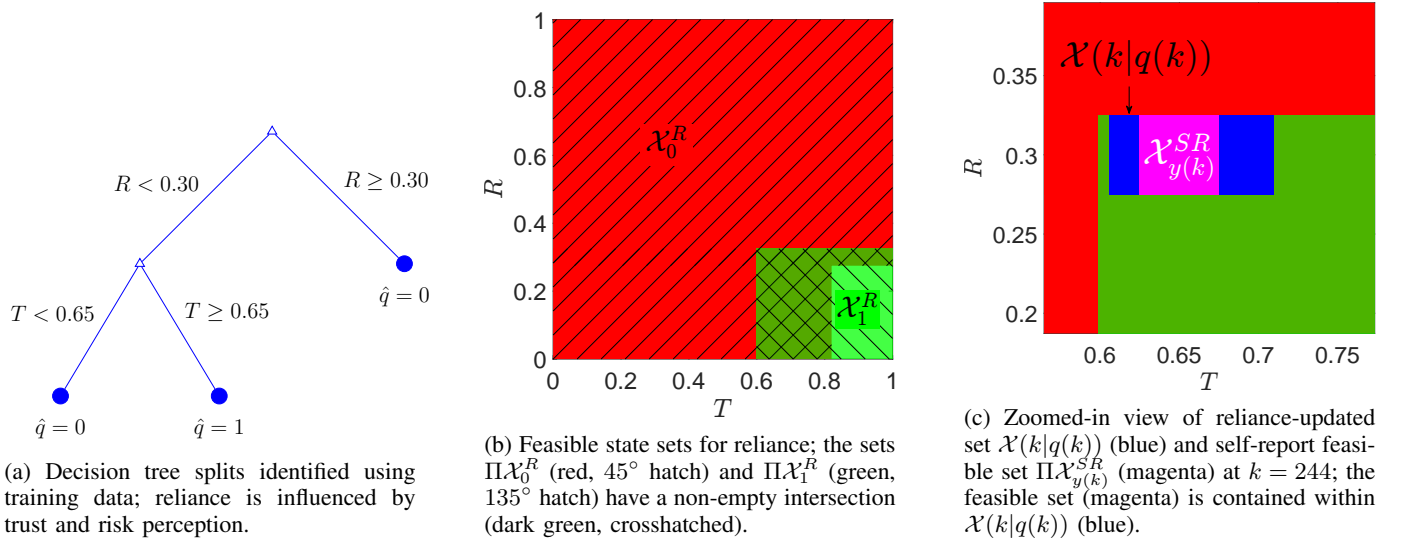

\begin{figure*}
    \centering
    \begin{subfigure}[b]{0.3\linewidth}
        \centering
       \includegraphics[width=\linewidth]{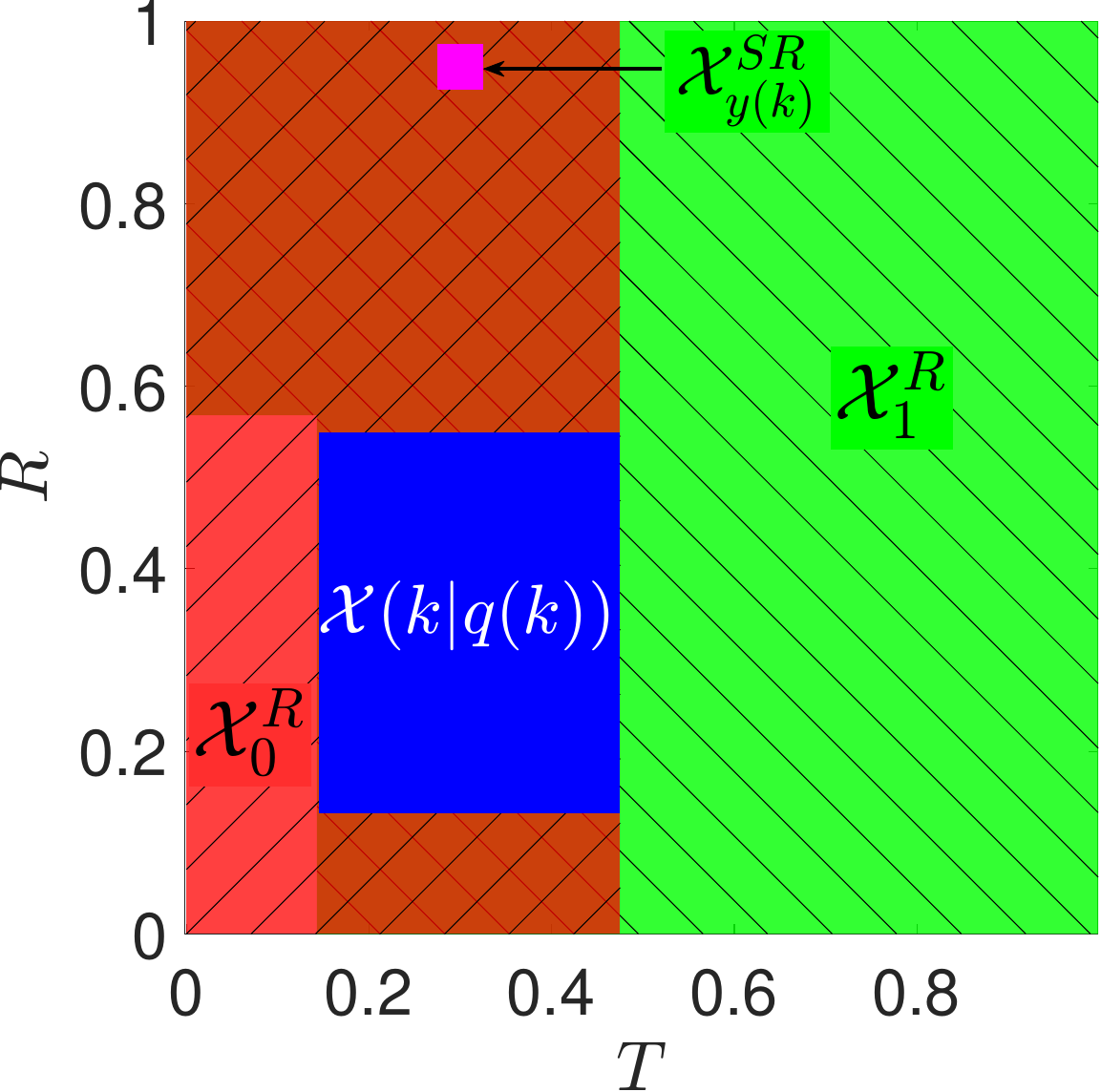}
        \caption{Visualization of failed containment of the feasible set $\Pi \mathcal{X}^{SR}_{y(k)}$ (magenta) by $\mathcal{X}(k|q(k))$ (blue) at $k=319$.}
        \label{fig:P004_sets}
    \end{subfigure}
    \hfill
    \begin{subfigure}[b]{0.6\linewidth}
        \centering
        \includegraphics[width=\linewidth]{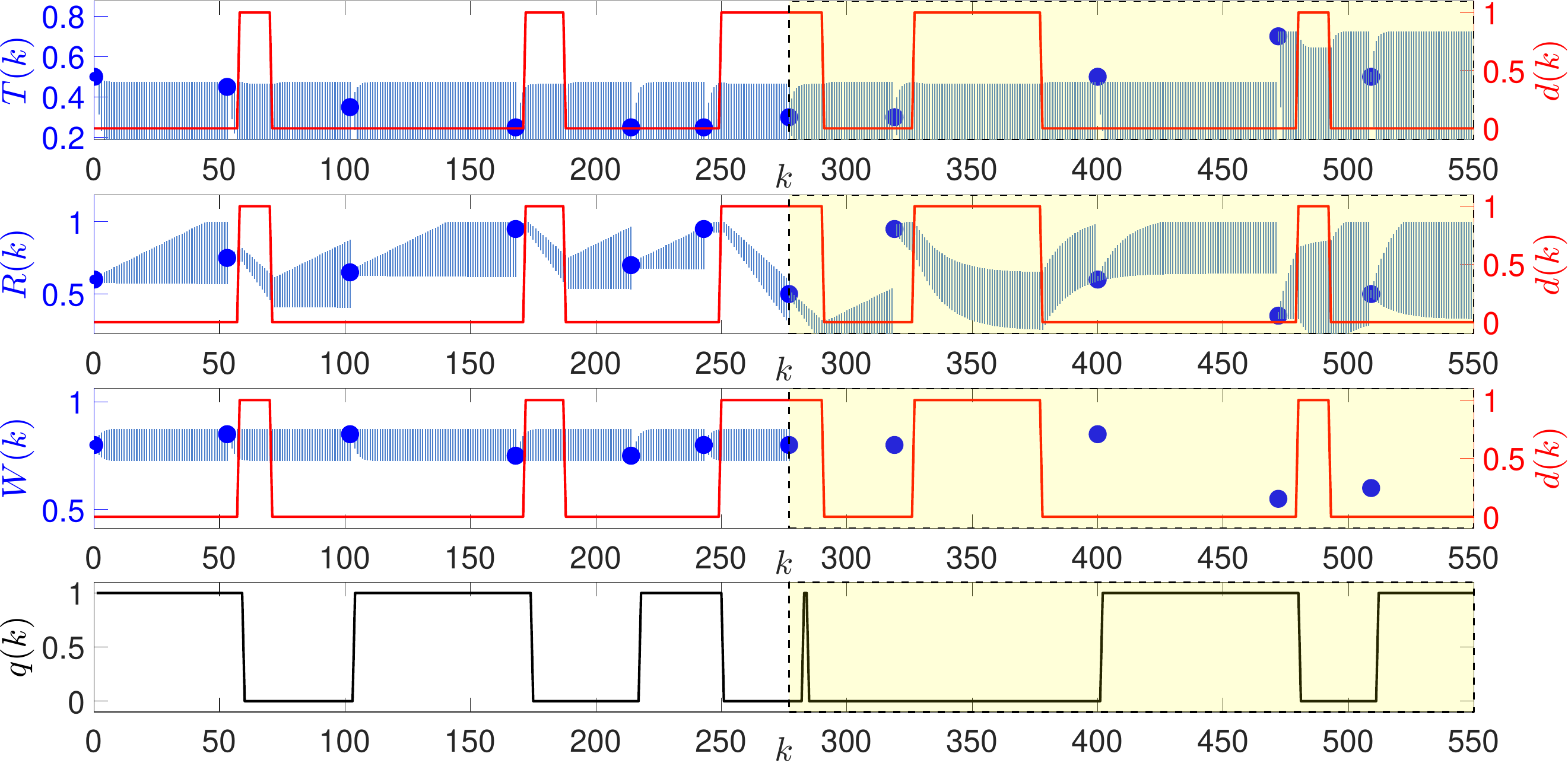}
        \caption{Time-series data collected for participant P04, along with the reachable sets projected onto each dimension (blue envelopes). Blue dots indicate self-reported cognitive states. The testing phase ($N_T\leq k\leq N$) is highlighted in yellow.}
        \label{fig:P004_trajectory}
    \end{subfigure}
    \caption{An illustration of failed reachset conformance during testing for participant P04 due to discrepancy between training and testing data.}
    \label{fig:P004_combined}
\end{figure*}

\begin{table}[t]
\centering
\caption{Consistency $O(\mathcal{X})$, uncertainty $\overline{\sigma}(\mathcal X)$, average model training time $\overline{t}_{\text{train}}$, and average set-computation time $\overline{t}_{\text{set}}$ of the set-valued state estimates across participants.}
\label{tab:estimation_performance}
\begin{tabular}{lcccc}
\toprule
Participant ID & $O(\mathcal{X})$ (\%) & $\overline{\sigma}(\mathcal{X})$ (\%) & $\overline{t}_{\text{train}}$ (s) & $\overline{t}_{\text{set}}$ (s) \\
\midrule
P01 & 100.00 & 0.42  & 7.736  & 0.057 \\
P03 & 75.00  & 12.47 & 6.983  & 0.035 \\
P04 & 50.00  & 18.91 & 9.331  & 0.054 \\
P05 & 100.00 & 4.23  & 8.640  & 0.060 \\
P06 & 75.00  & 5.05  & 10.189 & 0.067 \\
P07 & 100.00 & 33.94 & 5.734  & 0.037 \\
P08 & 80.00  & 24.98 & 5.871  & 0.033 \\
P09 & 75.00  & 11.57 & 7.644  & 0.059 \\
P10 & 100.00 & 38.37 & 7.280  & 0.038 \\
P12 & 75.00  & 5.29  & 8.586  & 0.063 \\
P13 & 75.00  & 2.79  & 8.098  & 0.050 \\
P15 & 75.00  & 6.69  & 6.428  & 0.042 \\
P18 & 66.67  & 2.95  & 7.306  & 0.050 \\
P20 & 50.00  & 7.77  & 9.097  & 0.056 \\
\bottomrule
\end{tabular}
\end{table}

\subsection{State Estimation Performance}

To assess the performance of the state estimator, we calculate
\begin{enumerate}
    \item the \textit{consistency} of the estimates, denoted by $O(\mathcal{X})$, as the percentage of time steps $k\in \mathcal K_{SR}$ for which the set of states consistent with the self-reports, $\Pi \mathcal X^{SR}_{y(k)}$, overlaps with the set-valued estimate, $\mathcal X(k \mid q(k))$, and
    \item the \textit{uncertainty} of the estimates, denoted by $\overline{\sigma}(\mathcal X)$, as the mean percentage of the subspace $\Pi \mathcal S$ covered by the estimates $\mathcal X(k\mid k)$ over $N_T\leq k\leq N$.
\end{enumerate}

Table~\ref{tab:estimation_performance} summarizes the consistency and uncertainty of the set-valued state estimates across participants, as well as the average model training and retraining (where applicable) times ($\overline{t}_{train}$) and the average set-computation times ($\overline{t}_{set}$). The estimator is $100\%$ consistent for four participants (P01, P05, P07, P10), and at least $75\%$ for 11 participants in total. The two participants with the lowest consistency (P04 and P20, both at $50\%$) are representative cases in which system matrices and noise bounds estimated from the training data were not representative of participants' testing behavior. Fig.~\ref{fig:P004_combined} illustrates failed reachset conformance during testing for participant P04. Fig.~\ref{fig:P004_sets} shows the reliance-updated set $\mathcal X(k\mid q(k))$ and $\Pi\mathcal X^{SR}_{y(k)}$ at $k=319$; the two sets do not overlap, such that $\mathcal X(k\mid q(k))\cap \Pi\mathcal X^{SR}_{y(k)}=\emptyset$. Fig.~\ref{fig:P004_trajectory} shows the time-series data for participant P04. Reachable sets (blue envelopes) are constructed for all states during the training phase ($0 \leq k < N_T$) via Eqns.~\ref{eq:set_predict_1}--\ref{eq:set_predict_2}, achieving reachset conformance for all cognitive states. During the testing phase ($N_T \leq k \leq N$, highlighted in yellow), the set-valued state estimates are projected onto the dimensions influencing reliance for this participant (trust and risk). At $k = 319$, the self-reported $R(k)$ falls outside the set estimate, indicating a conformance failure for risk perception. Consequently, data for $0\leq k\leq319$ is used to re-estimate the system matrices, noise bounds, and reliance regions to identify model parameters that enforce reachset conformance over $0\leq k\leq 319$. At $k=472$, conformance fails for both trust and risk perception; the re-estimated noise bounds enforcing reachset conformance on $0\leq k\leq 472$ are larger in magnitude, resulting in wider projected reachable sets for $k\geq 472$ in Fig.~\ref{fig:P004_trajectory}.

In Table~\ref{tab:estimation_performance}, uncertainty varies considerably across participants, ranging from $0.42\%$ to $38.37\%$, depending on the estimated process noise bounds; higher uncertainty reflects participants whose cognition is not well captured by the affine dynamical model. Among the 15 participants, uncertainty below $15\%$ is achieved for 10 of them. Notably, P01 achieves $100\%$ consistency with approximately $0.4\%$ average uncertainty, suggesting that the estimated sets reliably contain the true state while remaining tightly concentrated around it.  Model training and retraining times average between $7$--$8$ seconds, which is acceptable given that retraining is triggered only when necessary. Importantly, the average set-computation times range from $0.033$ to $0.067$ seconds across participants, well below the sampling interval ($1$ second), confirming that the estimator operates faster than real time and could be implemented online. All computations were carried out on a desktop PC with an Intel Core i7 (9th Gen) CPU.

\subsection{$N_p$-step-ahead predictions of reliance}
\begin{table}[t]
\centering
\caption{AUC-ROC comparison of $N_p$-step-ahead reliance predictions across methods.}
\label{tab:auc_roc}
\begin{tabular}{clccc}
\toprule
$N_p$ & Method & Mean & Std & Median \\
\midrule
\multirow{3}{*}{15} & Open-Loop            & 0.67 & 0.24 & 0.73 \\
                    & Particle Filter      & 0.72 & 0.22 & 0.76 \\
                    & Set-Valued Estimator & \textbf{0.74} & 0.24 & 0.75 \\
\midrule
\multirow{3}{*}{30} & Open-Loop            & 0.67 & 0.26 & 0.74 \\
                    & Particle Filter      & 0.71 & 0.24 & 0.76 \\
                    & Set-Valued Estimator & \textbf{0.74} & 0.24 & \textbf{0.76} \\
\midrule
\multirow{3}{*}{45} & Open-Loop            & 0.68 & 0.27 & 0.76 \\
                    & Particle Filter      & 0.74 & 0.25 & 0.80 \\
                    & Set-Valued Estimator & \textbf{0.80} & \textbf{0.20} & \textbf{0.90} \\
\midrule
\multirow{3}{*}{60} & Open-Loop            & 0.68 & 0.27 & 0.82 \\
                    & Particle Filter      & 0.76 & 0.25 & 0.80 \\
                    & Set-Valued Estimator & \textbf{0.82} & \textbf{0.21} & \textbf{0.87} \\
\bottomrule
\end{tabular}
\end{table}
In addition to its utility for cognitive state estimation at the current time step, our state estimation scheme can also be leveraged for more accurate predictions of reliance behavior for online decision-making.  To that end, we evaluate the utility of the state estimator for $N_p$-step-ahead predictions of reliance, $\hat P(q(k)=1\mid k{-}N_p)$ for $k>N_p$, bench-marked against two baselines: (1) an open-loop, model-based prediction $\hat P(q(k\mid N_T)=1)$ (using Eqn.~\ref{eq:affine_dyn} and decision tree estimated on $0\leq k\leq N_T$), and (2) a particle filter with Gaussian process and measurement noise serving as a probabilistic baseline (see Appendix~\ref{app:particle_filter}). The particle filter baseline is constructed using the same individually-identified system matrices, noise bounds, and decision tree partitions as the set-valued estimator, ensuring that differences in predictive performance are attributable to the state estimation mechanism rather than the model identification procedure. 

To predict reliance using the set-valued state estimate, we first compute the $N_p$-step-ahead reachable set $\mathcal X(k \mid k{-}N_p)$ recursively using Eqn.~\ref{eq:one_step_ahead}, starting from $\mathcal X(k{-}N_p\mid k{-}N_p)$. Next, we use the law of total probability to re-write $\hat P(q(k)=1\mid k{-}N_p)$ as
\(
     \sum_{i=1}^{\hat M} \left[ \hat P(q=1\mid x\in \hat{\mathcal S}_i)\cdot \hat P(x(k)\in \hat{\mathcal S}_i\mid k{-}N_p) \right] \enspace.
\)
The first term in each product is given by the decision tree. The set-valued estimator makes no distributional assumptions on $x(k)$; however, for the sake of predicting reliance, computing the second term requires a probability, which we estimate by assuming a uniform distribution over the set $\mathcal B(\mathcal X(k\mid k{-}N_p))$, giving
\begin{align*}
\hat P(x(k)\in \hat{\mathcal S}_i\mid k{-}N_p) = \frac{\text{vol}(\mathcal B(\mathcal X(k\mid k{-}N_p))\cap\hat{\mathcal S}_i)}{\text{vol}(\mathcal B(\mathcal X(k\mid k{-}N_p)))}\enspace.
\end{align*}

The predictive performance is evaluated by computing the Area Under the Receiver Operating Characteristic Curve (AUC-ROC) using the predicted probabilities and the true reliance labels $q(k)$. Table~\ref{tab:auc_roc} compares the AUC-ROC values for the set-valued state estimator against the baseline cases for several prediction horizon lengths: $N_p\in \{15,30, 45, 60\}$. The set-valued estimator consistently achieves the highest mean AUC-ROC across all prediction horizons $N_p$, with performance improving as $N_p$ increases (mean AUC of 0.74, 0.74, 0.80, and 0.82, respectively). At the shortest horizon $N_p = 15$, all three methods perform comparably (mean AUC of 0.67, 0.72, and 0.74), suggesting that short-horizon predictions are less sensitive to the quality of the state estimate. In other words, the dynamic model alone provides sufficient predictive power over a small number of steps. For larger $N_p$, however, the mean prediction accuracy of the set-valued state estimator differs from the open loop accuracy by more than $0.1$; this suggests that the recursive intersection of past reliance feasible sets $\mathcal{X}^R_{q(k-m)}$ provides increasingly informative constraints over longer horizons, effectively leveraging the history of reliance observations to improve the state estimate. In contrast, the open-loop baseline, which propagates the model forward from $k=N_T$ without any measurement-based updates, remains largely constant across all prediction horizons (mean AUC $\approx 0.68$), confirming that the performance improvements of the set-valued estimator are attributable to the measurement-based updates rather than the underlying dynamical model alone. The particle filter also improves with $N_p$ (mean AUC of 0.72 to 0.76), but is consistently outperformed by the set-valued estimator in mean AUC. The set-valued estimator also achieves a lower standard deviation at $N_p \in \{45, 60\}$ (0.20 and 0.21 vs.\ 0.25 for the particle filter). The degraded performance of the particle filter at longer horizons suggests that errors due to ad hoc (Gaussian) characterization of the process and measurement noise distribution can compound over time, leading to larger prediction errors as $N_p$ grows. 

\section{Conclusion}\label{sec:conclusion}
In this work we presented the first set-based framework for estimating human cognitive states in a continuous, non-trial-based human-automation interaction setting. Unlike probabilistic approaches dominant in the human-automation interaction literature, the proposed framework treats process and measurement uncertainties as unknown but bounded, avoiding the need for large structured datasets or distributional assumptions on noise. Reachset conformance is enforced during system identification to systematically estimate noise bounds, and binary reliance observations and intermittent, quantized self-reports are fused online to produce set-valued estimates of the cognitive states influencing reliance. An online model update mechanism further enables the estimator to adapt when the initially identified model fails to capture behaviors observed during testing.

The framework was evaluated in an in-person driving simulator experiment, where trust, perceived risk, and workload were estimated for individual participants during continuous interaction with SAE Level 3 automation. The set-valued estimator achieved at least 75\% consistency for $\approx 80\%$ participants, and $N_p$-step-ahead reliance predictions outperformed both an open-loop baseline and a particle filter across all prediction horizons $N_p \in \{15, 30, 45, 60\}$. The performance advantage over the particle filter increased with $N_p$, suggesting that the recursive intersection of past reliance feasible sets provides increasingly informative constraints over longer horizons, without requiring distributional assumptions that can compound errors over time.

Several directions remain open for future work. First, the estimator should be deployed within a closed-loop trust calibration or adaptive automation policy to directly assess its utility for meeting control objectives. Second, the current framework estimates only the subset of cognitive state dimensions predictive of reliance, since reliance is the only continuously available measurement. However, incorporating additional measurements sensitive to all state dimensions would enable full-state estimation. Finally, extending the framework to leverage population-level priors for model parameters could improve identification for individuals whose training data is limited.

\section*{Acknowledgments}
We thank Xipeng Wang for developing the driving simulator software and Tyler Hsieh for his assistance with experiment ideation.

\appendices
\section{Simplifying the Reachset-Conformance Constraints}\label{app:reachset_constraints}
In the optimization problem described by Eqn.~\ref{eq:sys_ID_bounds} in the main text, the diagonal system matrix $\hat{A}$ and axis-aligned intervals $\mathcal{W}$ and $\mathcal{V}$ eliminate the need for a toolbox for reachability analysis (which can be computationally expensive). Instead, we exploit the interval arithmetic to analytically compute forward reachable sets and evaluate the set containment constraints. Since $\hat{A}$ is diagonal and $\mathcal{W}$, $\mathcal{V}$ are axis-aligned intervals, the reachable set at each time step remains an axis-aligned interval, whose center $\mu$ and half-width $P$ can be propagated analytically. At each time step $k$, the reachable interval center and half-width are given by
\begin{align*}
    \mu(k) &= \hat{A}\mu(k{-}1) + \hat{B}d(k{-}1) + \hat{C}, \\
    P(k)   &= \hat{A}P(k{-}1) + \epsilon_w\enspace.
\end{align*}
The interval is then clipped to the state domain $\mathcal{S} = [0,1]^3$ via
\begin{align*}
    I_{\min} &= \max\!\left(0,\;\mu(k) - P(k)\right)\enspace, \\
    I_{\max} &= \min\!\left(1,\;\mu(k) + P(k)\right)\enspace, \\
    \mu(k)   &= \frac{I_{\min} + I_{\max}}{2}, \qquad
    P(k)      = \frac{I_{\max} - I_{\min}}{2}\enspace.
\end{align*}
Whenever a self-report $y(k)$ is available, the measurement interval defined by $\mathcal X_{y(k)}^{SR}$ is likewise clipped to $[0,1]^3$. 
Let $\mu_y(k)$ and $P_y(k)$ denote the center and half-width of the clipped measurement interval, respectively. Enforcing containment of the self-report
within the reachable interval yields the inequality constraints
\begin{align*}
    \bigl(\mu_y(k) + P_y(k)\bigr) - \bigl(\mu(k) + P(k)\bigr) &\leq 0\enspace, \\
    \bigl(\mu(k)   - P(k)  \bigr) - \bigl(\mu_y(k) - P_y(k)\bigr) &\leq 0\enspace.
\end{align*}
Finally, once the containment constraints have been evaluated, the reachable interval 
is reset to the self-report measurement, such that
\begin{align*}
    \mu(k) = \mu_y(k), \qquad P(k) = P_y(k)\enspace.
\end{align*}

\section{Probabilistic Baseline: Particle Filter}\label{app:particle_filter}
We use the estimated system matrices, noise bounds, and reliance regions described in 
the main text to construct a probabilistic baseline estimation scheme via a particle filter. The particle 
filter maintains $N_{PF} = 1000$ particles $\{x^{(i)}\}_{i=1}^{N_{PF}}$, each representing 
a hypothesis of the projected cognitive state vector $\Pi x$, consisting only of the 
state dimensions in $\mathcal{J}$ identified as predictive of reliance by the decision tree.

\paragraph{Initialization}
Particles are initialized around the first available self-report $y(N_T)$, perturbed 
by noise drawn from a Gaussian distribution with standard deviation $\epsilon_v$, such that for $i = 1, \ldots, N_{PF}$,
\begin{align*}
    x^{(i)}(N_T\mid N_T) \sim \Pi y(N_T) + \mathcal{N}(0,\, \mathrm{diag}(\epsilon_v^2))\enspace,
\end{align*}
and clipped to the valid domain $\Pi\mathcal{S}$. Each particle is assigned a likelihood 
initialized uniformly, such that
\begin{align*}
    \lambda^{(i)}(N_T\mid N_T) = \frac{1}{N_{PF}}, \quad i = 1, \ldots, N_{PF}\enspace.
\end{align*}

\paragraph{Prediction}
At each time step $k$, particles are propagated through the identified system dynamics 
with additive process noise $w^{(i)}(k) \sim \mathcal{N}(0,\,\mathrm{diag}(\epsilon_w^2))$, 
such that
\begin{align*}
    x^{(i)}&(k\mid k{-}1) = \\
    &\Pi\left(\hat{A}x^{(i)}(k{-}1\mid k{-}1) + \hat{B}d(k{-}1) 
    + \hat{C} + w^{(i)}(k)\right)\enspace,
\end{align*}
and clipped to $\Pi\mathcal{S}$. The likelihoods are carried forward unchanged, 
i.e., $\lambda^{(i)}(k\mid k{-}1) = \lambda^{(i)}(k{-}1\mid k{-}1)$.

\paragraph{Reliance-Based Update}
The observed reliance label $q(k) \in \{0, 1\}$ is used to update the particle 
likelihoods via the decision tree. For each particle $x^{(i)}(k\mid k{-}1)$, the 
decision tree returns a probability $P(q(k) \mid x^{(i)}(k\mid k{-}1))$, and the 
likelihoods are updated according to
\begin{align*}
    \tilde{\lambda}^{(i)}(k\mid q(k)) = \lambda^{(i)}(k\mid k{-}1) 
    \cdot P\!\left(q(k) \mid x^{(i)}(k\mid k{-}1)\right)\enspace,
\end{align*}
and re-normalized to form a valid probability distribution over particles.

\paragraph{Self-Report-Based Update}
Whenever a self-report $y(k)$ is available, the likelihoods are further updated 
using a Gaussian likelihood with standard deviation $\epsilon_v$ per state dimension, 
such that
\begin{align*}
    \tilde{\lambda}^{(i)}(k\mid &q(k), y(k)) = \lambda^{(i)}(k\mid q(k)) \times 
    \\&\exp\!\left(-\frac{1}{2}\sum_{j}\frac{\left(x^{(i)}_j(k\mid k{-}1) - 
    y_j(k)\right)^2}{\epsilon_{v,j}^2}\right)\enspace,
\end{align*}
and re-normalized as
\begin{align*}
    \lambda^{(i)}(k\mid k) = \frac{\tilde{\lambda}^{(i)}(k\mid q(k), y(k))}
    {\sum_{j=1}^{N_{PF}}\tilde{\lambda}^{(j)}(k\mid q(k), y(k))}\enspace.
\end{align*}
When no self-report is available, we set 
$\lambda^{(i)}(k\mid k) = \lambda^{(i)}(k \mid q(k))$.

\paragraph{Resampling}
Over time, likelihood degeneracy can occur, whereby most particles are assigned 
negligible likelihood and only a few dominate the distribution. To detect this, we 
compute the effective sample size
\begin{align*}
    N_{\mathrm{eff}}(k) = \frac{1}{\sum_{i=1}^{N_{PF}} 
    \left(\lambda^{(i)}(k\mid k)\right)^2}\enspace,
\end{align*}
which ranges from $1$ (full degeneracy) to $N_{PF}$ (uniform likelihoods). Whenever 
$N_{\mathrm{eff}}(k) < N_{PF}/2$, systematic resampling is applied~\cite{Kitagawa01031996}. A single random 
offset $u_{\mathrm{off}} \sim \mathcal{U}(0, 1/N_{PF})$ is drawn, and $N_{PF}$ evenly 
spaced thresholds $\{u_{\mathrm{off}} + (j-1)/N_{PF}\}_{j=1}^{N_{PF}}$ are swept 
across the cumulative likelihood distribution to select particles, replicating 
high-likelihood particles and discarding low-likelihood ones. After resampling, 
likelihoods are reset uniformly to $\lambda^{(i)}(k\mid k) = 1/N_{PF}$.

\paragraph{$N_p$-Step-Ahead Reliance Prediction}
At each time step $k>N_p$, the current particle population is propagated $N_p$ 
steps forward using only the prediction step (no updates), yielding
\begin{align*}
    &x^{(i)}(k\mid k{-}N_p) = \Pi(\hat{A}^{N_p}x^{(i)}(k{-}N_p\mid k{-}N_p) 
    \\&+ \sum_{j=0}^{N_p-1}\hat{A}^{j}(\hat{B}d(k+N_p-1-j) 
    + \hat{C}) + w^{(i)})\enspace.
\end{align*}
The decision tree is then queried on the propagated particles to produce a 
likelihood-weighted probability of reliance $\hat{P}\!\left(q(k+N_{\mathrm{rel}})=1\right)$, computed as
\begin{align*}
     \sum_{i=1}^{N_{PF}} 
    \lambda^{(i)}(k\mid k)\cdot \hat P\!\left(q=1 \mid x^{(i)}(k+N_{\mathrm{rel}}\mid k)\right)\enspace.
\end{align*}

\bibliography{references}
\bibliographystyle{IEEEtran}


\section*{Biography Section}

\begin{IEEEbiography}[{\includegraphics[width=1in,height=1.25in,clip,keepaspectratio]{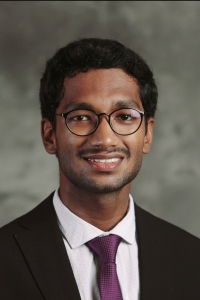}}]{Sibibalan Jeevanandam}
 is a fourth year doctoral student at Purdue University. He earned his Bachelors of Technology in Mechanical Engineering from the Indian Institute of Technology Tirupati (India) in 2022. His research interests are control-oriented modeling and estimation, set-based methods, and human-machine interaction. He received the Winkelman Fellowship from the School of Mechanical Engineering in 2022, and the Summer Undergraduate Research Fellowship Graduate Mentor Award in 2023, both at Purdue.
\end{IEEEbiography}

\begin{IEEEbiography}[{\includegraphics[width=1in,height=1.25in,clip,keepaspectratio]{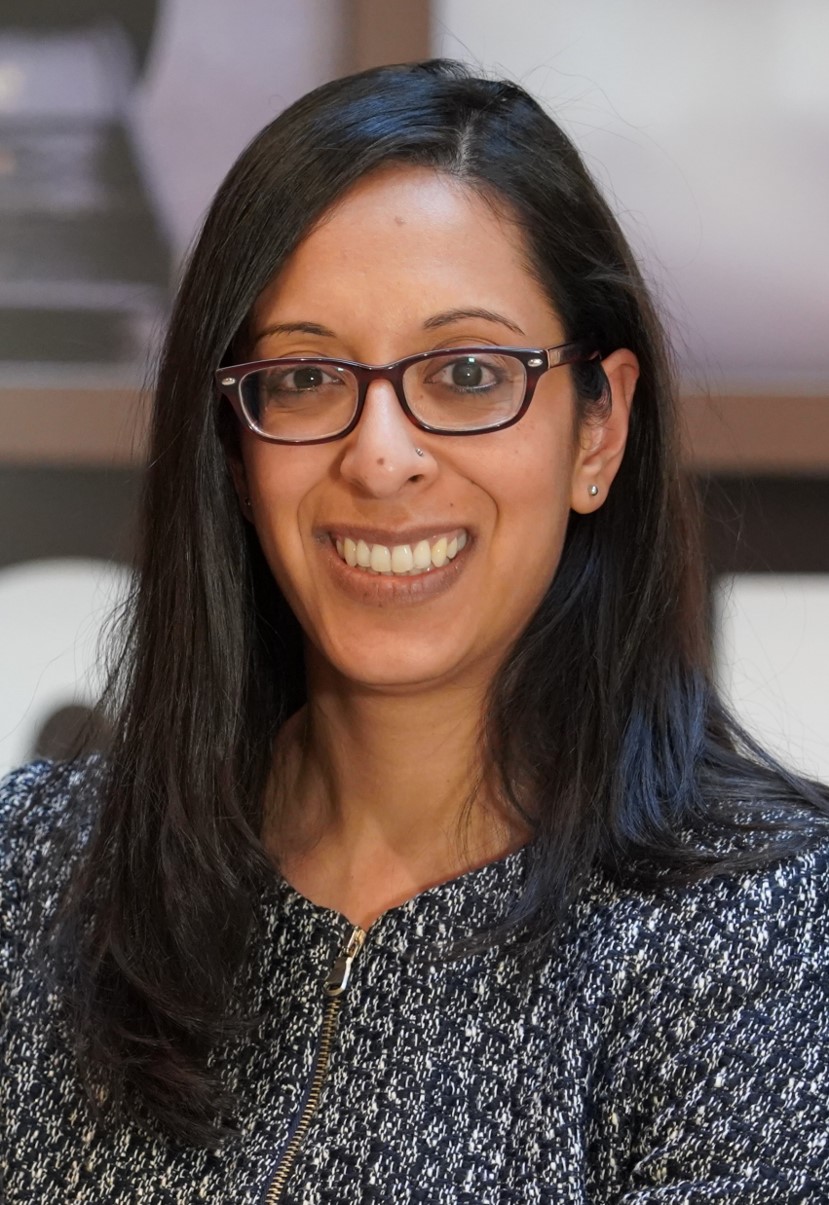}}]{Neera Jain}
 is an Associate Professor in the School of Mechanical Engineering at Purdue University, West Lafayette, IN, USA. Dr. Jain earned her M.S. and Ph.D. degrees in mechanical engineering from the University of Illinois at Urbana-Champaign in 2009 and 2013, respectively. She earned her S.B. from the Massachusetts Institute of Technology in 2006. Her research interests include dynamic modeling, state estimation, and optimal control with applications primarily to human-autonomy teaming and thermal management in vehicles and buildings. She has held visiting research positions at Mitsubishi Electric Research Laboratories, Air Force Research Laboratory (AFRL) at Wright-Patterson Air Force Base (AFB) (Aerospace Systems Directorate) and AFRL at Kirtland AFB (Space Vehicles Directorate). She is the recipient of the 2022 NSF CAREER Award, a National Research Council (NRC) Research Associateship Award, the 2023 ASME Dynamic Systems and Control Division (DSCD) Outstanding Young Investigator Award, and the 2024 ASME DSCD Rudolf Kalman Best Paper Award.
\end{IEEEbiography}

\vfill

\end{document}